\documentclass[12pt,peerreview,onecolumn]{IEEEtranTCOM}
\pdfoutput=1
\usepackage{amsmath,epsfig,amssymb,verbatim,amsopn,subfigure,color}
\usepackage{cite,xspace}
\usepackage{array,algorithm,algorithmic}
\usepackage{bbm}
\usepackage{array}
\usepackage{multirow}
\usepackage{multicol}
\usepackage{rotating}
\usepackage{dsfont,float}
\usepackage{stfloats}
\usepackage{enumerate,makecell}
\normalsize

\ifCLASSOPTIONpeerreview
\renewcommand{\baselinestretch}{1.78}
\fi

\ifCLASSOPTIONonecolumn
    
\else
    
\fi

\newcommand{\redcom}[1]{{\color{black}#1}\xspace}

\makeatletter
\let\l@ENGLISH\l@english
\makeatother

\makeatletter
\renewcommand*{\@opargbegintheorem}[3]{\trivlist
  \item[\hskip \labelsep{\itshape #1\ #2}] {\itshape (#3):} {\normalfont}}
\makeatother
\usepackage{relsize}
\newcommand{\BSdensity}  {\lambda_{B}}
\newcommand{\adensity} {\lambda_{a}}
\newcommand{\psuc} {\bar{p}_{\textrm{suc}}}
\newcommand{\achannel} {\bar{p}_{\textrm{utility}}}
\newcommand{\pndrop} {p_{\textrm{nodrop}}}
\newcommand{\apndrop} {\bar{p}_{\textrm{nodrop}}}
\newcommand{\pouta} {p_{\textrm{suc}1}}
\newcommand{\apouta} {\bar{p}_{\textrm{suc}1}}
\newcommand{\poutb} {p_{\textrm{suc}2}}
\newcommand{\apoutb} {\bar{p}_{\textrm{suc}2}}
\newcommand{\ii}{\textbf{i}}

\graphicspath{{figs/},{Figures/}}

\newtheorem{lemma}{Lemma}
\newtheorem{remark}{Remark}
\newtheorem{theorem}{Theorem}
\newtheorem{proposition}{Proposition}
\newtheorem{corollary}{Corollary}
\graphicspath{{figs/},{Figures/}}

\ifCLASSOPTIONpeerreview
\else
\fi
\newcommand{\AuthorOne}{Jing~Guo, {\em{Student Member, IEEE}}}
\newcommand{\AuthorTwo}{Salman~Durrani, {\em{Senior Member, IEEE}}}
\newcommand{\AuthorFour}{Halim~Yanikomeroglu, {\em{Fellow, IEEE}}}
\newcommand{\AuthorThree}{Xiangyun~Zhou, {\em{Member, IEEE}}}
\newcommand{\ThankOne}{J. Guo, S. Durrani and X. Zhou are with the Research School of Engineering, The Australian National University, Canberra, ACT 2601, Australia (Emails: \{jing.guo, salman.durrani, xiangyun.zhou\}@anu.edu.au). H. Yanikomeroglu is with the Department of Systems and Computer Engineering, Carleton University, Ottawa, ON K1S 5B6, Canada (E-mail: halim@sce.carleton.ca).}

\begin{document}

\title{Massive Machine Type Communication with Data Aggregation and Resource Scheduling}
\author{\IEEEauthorblockN{\AuthorOne,~\AuthorTwo,~\AuthorThree,~and~\AuthorFour\thanks{\ThankOne}}}
\maketitle

\vspace{+5mm}

\begin{abstract}
 To enable massive machine type communication (mMTC), data aggregation is a promising approach to reduce the congestion caused by a massive number of machine type devices (MTDs). In this work, we consider a two-phase cellular-based mMTC network where MTDs transmit to aggregators (i.e., aggregation phase) and the aggregated data is then relayed to base stations (i.e., relaying phase). Due to the limited resources, the aggregators not only aggregate data, but also schedule resources among MTDs. We consider two scheduling schemes: random resource scheduling (RRS) and channel-aware resource scheduling (CRS). By leveraging the stochastic geometry, we present a tractable analytical framework to investigate the signal-to-interference ratio (SIR) for each phase, thereby computing the MTD success probability, the average number of successful MTDs and probability of successful channel utilization, which are the key metrics characterizing the overall mMTC performance. Our numerical results show that, although the CRS outperforms the RRS in terms of SIR at the aggregation phase, the simpler RRS has almost the same performance as the CRS for most cases with regards to the overall mMTC performance. Furthermore, the provision of more resources at the aggregation phase is not always beneficial to the mMTC performance.
\end{abstract}
%

\ifCLASSOPTIONpeerreview
    \newpage
\fi

\section{Introduction}
Machine type communication (MTC) is envisaged to play a key role within future fifth generation networks~\cite{6231296}. With MTC, devices (i.e., smart meters, sensor nodes, appliances, etc.) can automatically ``talk" to each other without the human intervention, which will undoubtedly create unprecedented applications and new business models. While MTC brings great opportunities, it also poses significant challenges because of the distinct system requirements. Depending on the two major challenges, MTC can be classified into two types: \emph{massive machine type communication} (mMTC) and \emph{ultra-reliable machine type communication} (uMTC)~\cite{Bockelmann-2016}. mMTC is expected to provide massive access to a large number of often low-complexity and low-power machine type devices (MTDs). uMTC is expected to provide the network services for those MTDs with critical requirements in terms of latency and reliability. In this work, we will focus on the mMTC.

For mMTC, one way to handle a massive number of simultaneous device connections is to enhance the operation of random access channel of long term evolution and long term evolution advanced such that congestion and overloading can be reduced~\cite{Laya-2014}. Different methods have been proposed to offer more efficient access in the literature, such as the access class barring~\cite{3gpp}, prioritized random access~\cite{6720118} and backoff adjustment scheme~\cite{3gpp-2,6399192}. Another promising way to deal with the massive connection problem is the concept of data aggregation~\cite{6231296,Dawy-2015}. With the data aggregator, instead of the direct communication between MTDs and the core network (i.e., base station in cellular network), the traffic from MTDs is first transmitted to the designated data aggregator and the aggregator then relays the collected packets to the core network. Such an aggregation structure reduces the number of connections to the core networks thereby reducing congestion~\cite{6231296}. It also reduces the power consumption at the MTD side, since transmission link for MTD is greatly shortened~\cite{Dawy-2015}.

Recently, some papers have considered data aggregation in mMTC. For example, the MTD clustering problem was analyzed in both~\cite{6175028} and~\cite{Miao-2016}. The authors in~\cite{6175028} utilized the joint massive access control, while resource allocation to perform MTD grouping and an energy-efficiency cluster-head selection scheme was developed in~\cite{Miao-2016} to maximize the network life. The joint-user decoding was studied in~\cite{Pratas-2015} and a closed-form expression of the maximum zero-outage downlink rate was derived when multiple MTDs are attached to a cellular user. By employing the trunked radio system, the authors in~\cite{7248779} investigated the basic trade-off between latency and transmit power for delivering the aggregated traffic. These works, however, considered either a single aggregator or a single base station (BS) scenario and ignored the coverage nature of wireless transmission.

In a large scale cellular network with multiple aggregators and multiple BSs, interference exists and can affect the machine type communication when multiple aggregators (or BSs) share the same resource. There are only a few papers characterizing the interference in mMTC with data aggregation. In~\cite{Kwon-2013}, the distribution of the signal-to-interference ratio (SIR) for sensor nodes was derived and this work only focused on the single-hop scenario where the data from sensor nodes was aggregated at the data collectors. A clustering geometry model for MTD locations was introduced in~\cite{halim-2016}, where a two-hop clustering method based on slotted ALOHA was proposed and the analysis was based on simulations. The authors in~\cite{Malak-2016} analyzed the SIR coverage, rate coverage and energy consumption for a large-scale hierarchical wireless network with mMTC. Note that each device in~\cite{Malak-2016} was assumed to be guaranteed with an equal resource. However, when the spectrum resources are limited, it may be necessary to implement resource scheduling~\cite{6687313,7725919}. Furthermore, as envisioned in~\cite{Dawy-2015}, the aggregator should not be restricted to act as a simple relay and more intelligence at the aggregator (e.g., resource scheduling, spectrum sharing and so on) should help to improve the network efficiency.

Very recently, some papers have investigated the data aggregation for mMTC in the context of resource scheduling. Specifically, a resource allocation approach based on the consideration of the useful information content for individual MTDs was developed in~\cite{6504002}. By exploiting queueing-dynamics, a periodic scheduling algorithm was proposed in~\cite{6477828}. In~\cite{7417719}, the authors solved a resource sharing problem using a novel interference-aware bipartite graph. A low complexity optimal packet scheduler was developed in~\cite{Kumar-2016} using an iterative process. Again, these works considered a single core network and most of them did not take interference into account.

In this paper, we aim to characterize the interference and coverage performance for the mMTC with data aggregation in a large-scale cellular network system, where the resource scheduling scheme is implemented at the aggregator side. We leverage stochastic geometry, as a powerful math tools, to provide tractable analytical results~\cite{Haenggi-2012,jeffrey-2011,Guo-2016}. We make the following major contributions in this paper:
\begin{itemize}
  \item We introduce a tractable two-phase network model for massive machine type communication, where MTDs first transmit to their serving aggregators (aggregation phase) and then the aggregated data is delivered to BSs (relaying phase). We develop a general analytical framework to investigate the signal-to-interference ratio for both transmission phases. Based on this, we obtain the approximated yet accurate results for the MTD success probability, average number of successful MTDs and probability of successful channel utilization, which are key metrics to evaluate the overall mMTC performance. Compared to time consuming Monte-Carlo simulations, our derived analytical expressions allow for fast computation.

  \item  By taking into account the limited resources, we include two resource scheduling schemes implemented at the aggregators in the system: i) random resource scheduling (RRS) scheme, where aggregators randomly allocate the limited resources to MTDs; and ii) channel-aware resource scheduling scheme (CRS) scheme, where aggregators allocate resources to the MTDs having better channel conditions. The implementation of CRS is generally more complex than the RRS since it requires the perfect channel state information. Our results show that, even though the CRS scheme has much better SIR performance than the RRS scheme at the aggregation phase, when the resources at the aggregation phase are very limited, the RRS scheme only performs slightly worse than the CRS scheme in terms of the overall performance.

  \item Based on our derived results, we present different trade-offs in the system and investigate the effect of system parameters on the mMTC performance. Our results show that the provision of more resources at the aggregation phase is not always beneficial to the mMTC performance and sometimes it can even degrade the MTD success probability and average number of successful MTDs such as when the resources for the relaying phase are restricted, or the path-loss exponent is small, or the aggregators are very dense.
\end{itemize}
The remainder of the paper is organized as follows. Section~\ref{sec:system} presents the detailed system model and assumptions. The considered performance metrics and their descriptions are presented in Section~\ref{sec:metric}. The analysis for the aggregation phase and relaying phase is presented in Section~\ref{sec:firstphase} and Section~\ref{sec:secondphase}, respectively. The analytical results for the metrics are summarized in Section~\ref{sec:summary:metric}. Numerical and simulation results to study the machine type communication are discussed in Section~\ref{sec:result}. Finally, conclusions are presented in Section~\ref{sec:conclusion}.
\section{System Model}\label{sec:system}
\subsection{Network Model}\label{subsec:system:network}
Consider a single-tier cellular network where the location of base stations is modeled as a homogeneous Poisson point process (HPPP), denoted as $\Phi_{B}$, with density $\BSdensity$ in $\mathbb{R}^2$. We assume that the cellular network is overlaid with data aggregators, which are also spatially distributed according to an independent HPPP, denoted as $\Phi_{a}$, with density $\adensity$. Each aggregator is assumed to have a serving zone that is a disk region centered at the aggregator with radius $R_s$. Inside each serving zone, there are multiple machine type devices located around the aggregator. Generally, MTDs are either static or have low mobility. However, since the frequency of data transmission for each MTD is low, we assume that the location of MTDs requiring data transmission is modeled as a Mat\'ern cluster point process and the multiple aggregators constitute the parent point process of this Mat\'ern cluster point process\redcom{\footnote{\redcom{Our proposed model is applicable to smart utility metering and industry automation use cases for mMTC over cellular~\cite{Dawy-2015}, where the authorized MTDs are either static or have low mobility and are served by their designated aggregators. Hence, it is appropriate to model their locations as a Mat\'ern cluster point process. Note that a different MTD deployment has been considered in~\cite{Malak-2016}, which is an appropriate model when MTDs are scattered around without any concentration in their locations. The analytical framework developed in this paper can be applied to the network model in~\cite{Malak-2016}.}}}. In this way, for each cluster formed by the aggregator, the location of its MTDs \redcom{to be served} follows a PPP, where each MTD is uniformly and independently distributed with the distance distribution between a MTD and its serving aggregator being $f(r_m)=\frac{2r_m}{R_s^2}$. Let $\bar{m}$ denote the average number of MTDs in each aggregator's serving zone. The density of MTDs is then given by $\lambda_{m}=\bar{m}\lambda_{a}$~\cite{Haenggi-2012}, and the instantaneous number of MTDs in each aggregator, denoted as $K$, follows the Poisson distribution with mean $\bar{m}$, i.e., $\Pr(K=k)=\frac{1}{k!}\bar{m}^k\exp(-\bar{m})$~\cite{Kumar-2016}. Fig.~\ref{fig_network} illustrates a snapshot of the considered network model.
\vspace{-2mm}
\subsection{Transmission Model}\label{subsec:system:transmission}
We consider the uplink transmission model for machine type communication and the data transmission for a MTD can be divided into two phases\redcom{\footnote{\redcom{Similar to~\cite{Malak-2016,Pratas-2015,6504002,7417719,Kumar-2016,Azari-2016}, we do not model the random access in the network and the MTDs considered in the aggregation phase can be viewed as the MTDs that have been granted access to the aggregators.}}}. In the first phase (called the aggregation phase), the MTD tries to transmit its data with a fixed payload size $D$ to its serving aggregator\redcom{\footnote{\redcom{For simplicity, we assume that each aggregator has no buffer and transmits all its aggregated data in one go.}}}.Note that in our work the aggregator not only works as a relay, but can also implement the resource scheduling. We assume that, due to the spectrum limitation, there are $N$ orthogonal channels available that are scheduled for data transmission from MTDs to aggregators, and this channel set is denoted as $\mathcal{N}$. The limited channel resources $\mathcal{N}$ will be used among all aggregators across the entire network. Since at most one MTD is allowed to occupy one orthogonal channel within an aggregator's serving zone, there is only inter-cluster interference (i.e., the interference from MTDs in the serving zones of other aggregator) and no intra-cluster interference (i.e., the interference from MTDs within the serving zone of the same aggregator) exists.

\redcom{In this work, we consider a random resource scheduling (RRS) scheme, where any channel belonging to $\mathcal{N}$ will be independently and randomly allocated to any associated MTD with the same probability by the aggregator. As a benchmark, we also consider a channel-aware resource scheduling (CRS) scheme. Under this scheme, the MTDs with better fading (equivalently, better signal-to-noise ratio) will be preferentially assigned with the available channel resources~\cite{Dhillon-2014}. Note that both aggregators and BSs are assumed to have perfect channel state information (CSI) of their serving nodes. The RRS scheme does not need the CSI for the resource scheduling while the CRS scheme strongly relies on the CSI (i.e., the fading) for resource scheduling. We assume perfect CSI here in order to obtain benchmark results. The performance with imperfect CSI and practical schemes for obtaining the CSI are outside the scope of the present work.}
 \begin{figure}[t]
        \centering
        \vspace{-3mm}
        \includegraphics[width=0.6  \textwidth]{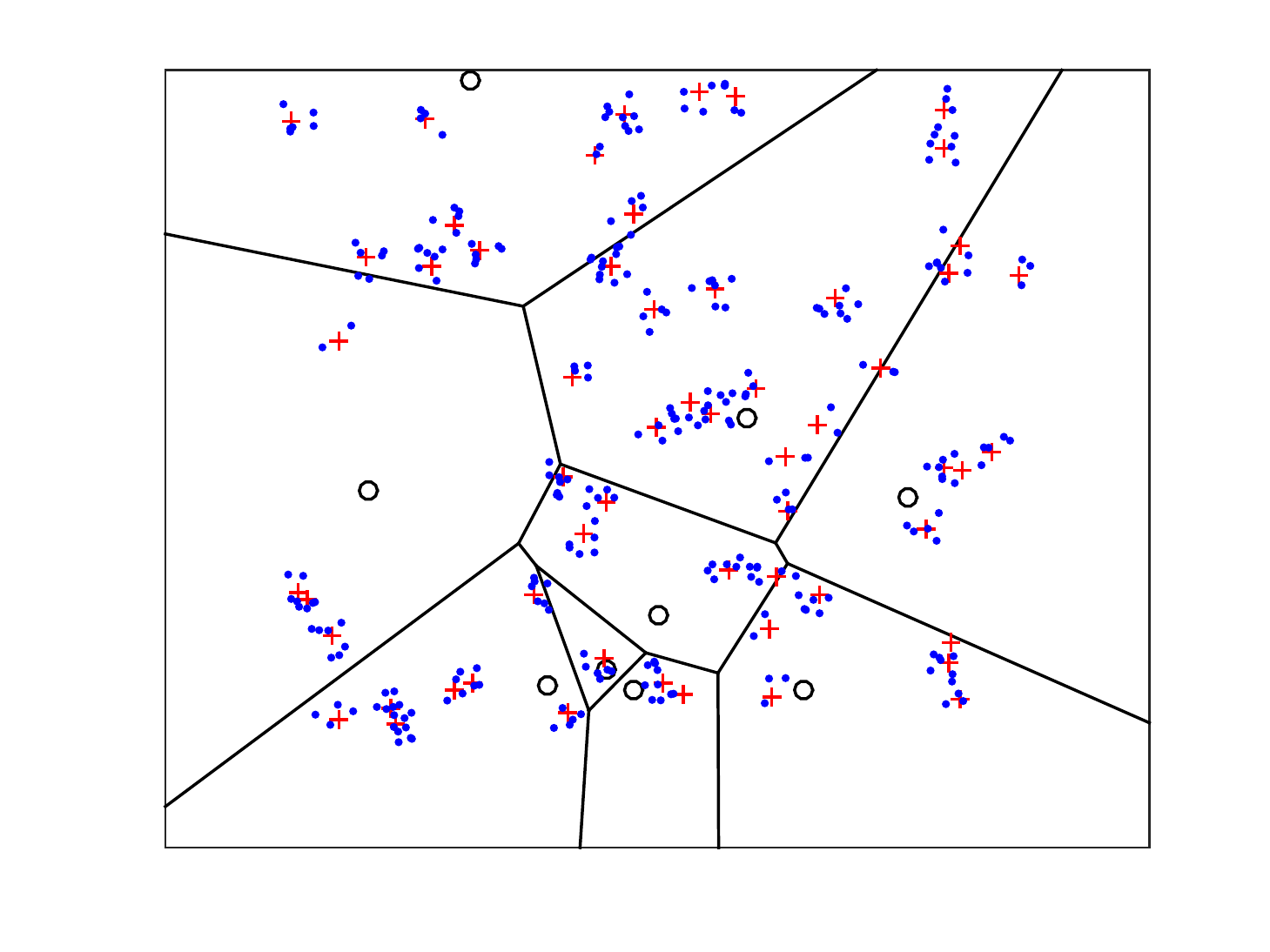}
        \vspace{-3mm}\caption{Illustration of the cellular network with massive machine type communications (${\color{black}\circ}=$ BS, ${\color{red}+}=$ aggregator, ${\color{blue}\cdot}=$ MTD).}
        \label{fig_network}
\end{figure}

 After aggregating data from the MTDs, each aggregator then \redcom{acts as an ordinary cellular user and} transmits the information to its closest BS in the second phase (called the relaying phase). In this way, the aggregators are partitioned by the Voronoi cell formed by the BSs, as shown in Fig.~\ref{fig_network}. Hence, the distance distribution between an aggregator and its associated BS is given by $f(r_a)=2\pi \BSdensity r_a\exp(-2\pi\BSdensity r_a^2)$~\cite{jeffrey-2011}. We further assume that the aggregator's density is far greater than the density of BSs such that each BS has at least one aggregator associated to it. For the multiple aggregators covered by the same BS, they access the BS using a round-robin fashion (i.e., the total available uplink resources are equally partitioned among these associated aggregators)~\cite{Singh-2015}. \redcom{Note that the relaying phase always begins after the aggregation phase is completed. We also assume that these phases occur synchronously in all aggregators.}

\subsection{Channel Model}
The path-loss plus block fading channel model is employed in this work. \redcom{The block fading assumption means that the fading is unchanged within one realization, but it changes independently from one realization to another realization}. The instantaneous received power at the receiver side is mathematically given by $p_{t}g r^{-\alpha}$, where $p_t$ is the transmit power from a transmitter, $g$ denotes the fading power gain on the transmission link, $r$ is the distance between the receiver and transmitter and $\alpha$ represents the path-loss exponent (i.e., $2<\alpha\leq6$). Throughout this work, we denote the fading experienced at the desired link for the aggregation and relaying phases as $h$ and $h'$, which are assumed to be the independently and identically distributed (i.i.d.) Nakagami-$m$ fading with integer $m_1$ and $m_2$, respectively. \redcom{Since the interfering nodes are more likely to be far away from the typical node, the fading on the interfering link is more likely to be severe. Hence we assume the fading on the interfering link to be i.i.d. Rayleigh fading, denoted as $g$.}

All the MTDs and aggregators are assumed to use full inversion power control with receiver sensitivity $\rho$~\cite{Malak-2016,Mohammad-2016}. For example, for a given MTD which is a distance $r_m$ away from its serving aggregator, its transmit power will be $\rho r_m^{\alpha}$. As we are considering an interference-limited scenario, the value of $\rho$ does not impact the performance of the network. Without loss of generality, we set $\rho$ to be unity in this work.

\section{Key Performance Metrics}\label{sec:metric}
Using the system model described above, we aim to investigate the network performance in terms of three metrics, namely the MTD success probability, average number of successful MTDs and probability of successful channel utilization. Their definitions, along with their mathematical formulations are described in this section below.

\subsection{MTD Success Probability}
This metric is the probability that the data sent by a MTD can be successfully received at the BS side after going through the aggregator. In order to guarantee that a typical MTD's transmission is successful, the following three conditions must be met:
 \begin{itemize}
  \item \textit{Condition 1: The typical MTD is not dropped by its serving aggregator.} Since the number of channels is limited for the aggregation phase, the MTD may not be assigned a channel such that its data is dropped.
  \item \textit{Condition 2: Given that the typical MTD is not dropped, it does not experience channel outage.} Due to the possible concurrent transmission of MTDs in other aggregators' serving zone on the same channel, the typical MTD will receive their generated inter-cluster interference. To ensure that the data can be successfully decoded at the aggregator, the signal-to-interference ratio at the aggregator (denoted as $\textsf{SIR}_1$) has to be greater than a certain threshold $\gamma_1$. Otherwise, this typical MTD is in channel outage.
  \item \textit{Condition 3: The serving aggregator is not in channel outage.} Similarly, when the aggregator transmits packet to its associated BS, it will receive the inter-cell interference from other aggregators. The SIR at the BS (denoted as $\textsf{SIR}_2$) has to be higher than a threshold $\gamma_2$.
 \end{itemize}
The MTD success probability can thus be generally written as
 \begin{align}\label{eq:general:MTDoutage}
 \psuc=\mathbb{E}\left\{\textbf{1}(\textrm{MTD is selected})\textbf{1}(\textsf{SIR}_{1}>\gamma_1)\textbf{1}(\textsf{SIR}_{2}>\gamma_2)\right\},
 \end{align}
\noindent where $\mathbb{E}\left\{\cdot\right\}$ is the expectation operator and $\textbf{1}(\cdot)$ is the indicator function.

The aggregation and relaying phases are dependent, \redcom{i.e., the aggregated data transmitted by an aggregator strongly relies on the instantaneous number of MTDs associated to the aggregator and the SIR performance on each channel}. For analytical tractability, we assume that each phase is independent~\cite{Malak-2016} and the accuracy of this assumption will be verified in Section~\ref{sec:result}. Thus, we decouple each factor in~\eqref{eq:general:MTDoutage} and rewrite the MTD success probability as\footnote{Note that for the CRS scheme, the formulation of $\psuc$ is slightly different from~\eqref{eq:summary_channel1}, because the first and second conditions are correlated. Its exact expression will be shown in Section~\ref{sec:summary:metric}.}
\begin{align}\label{eq:summary_channel1}
\psuc\approx&\,\apndrop\times\apouta\times\apoutb,
\end{align}
\noindent where the average non-drop probability $\apndrop$ is the average probability the a MTD is not dropped which corresponds to the first condition, $\apouta$ and $\apoutb$ denote the average channel success probability for the aggregation and relaying phases, which correspond to the second and third conditions, respectively.

 \subsection{Average Number of Successful MTDs}
 This metric evaluates, for a typical aggregator, the average number of its \redcom{served} MTDs whose data can be successfully received by the BS. It relies on two factors: i) the number of \redcom{served} MTDs whose data can be successfully decoded by the aggregator at the aggregation phase (equivalently, the number of active channels, denoted as $K_1$); and ii) whether the data aggregated at the aggregator can be successfully decoded by its associated BS. We can formally write this metric as
 \begin{align}\label{eq:generalnumbersuc}
 \bar{K}_{\textrm{suc}}&=\mathbb{E}\left\{K_1 \textbf{1}(\textsf{SIR}_{2}>\gamma_2)\right\}\nonumber\\
 &=\sum_{k_1=1}^{N}k_1\Pr(K_1=k_1)\poutb(k_1),
 \end{align}
 \noindent where $\Pr(K_1=k_1)$ is the probability mass function (PMF) of the number of active channels and $\poutb(k_1)$ is relaying phase's conditional channel success probability given $k_1$ active channels.

\textit{\underline{Rationale:}} The MTD success probability is a metric to assess the performance of a single MTD, while the average number of successful MTDs is a metric to assess the impact of aggregators. We have adapted their definitions from~\cite{7248779}. Both metrics are dependent, to a certain extent, on the available channel resources. In order to explicitly examine the efficiency of channel utilization, we consider the following metric, which evaluates the performance from the perspective of channel resources.
 \subsection{Probability of Successful Channel Utilization}
It is the average probability that a certain channel is occupied by a MTD and this MTD's data can finally be decoded by the BS. It can be written as
\ifCLASSOPTIONpeerreview
\begin{align}
 \achannel=\mathbb{E}\left\{\textbf{1}(\textrm{channel is occupied})\textbf{1}(\textsf{SIR}_{1}>\gamma_1)\textbf{1}(\textsf{SIR}_{2}>\gamma_2)\right\}.
\end{align}
\else
\begin{align}
 \achannel\!=\!\mathbb{E}\left\{\textbf{1}\!(\textrm{channel is occupied})\textbf{1}\!(\textsf{SIR}_{1}\!>\!\gamma_1)\textbf{1}\!(\textsf{SIR}_{2}\!>\!\gamma_2)\!\right\}.
\end{align}
\fi

Similar to the MTD success probability, by assuming the independence between the aggregation and relaying phases, we can have the probability of successful channel utilization as
\begin{align}\label{eq:summary_ulti}
 \achannel\approx &\,\bar{p}_{O}\times\apouta\times\apoutb,
\end{align}
\noindent where the average channel occupation probability $\bar{p}_{O}$ is the average probability that a channel is occupied by a MTD.

In the following two sections, we discuss the key elements determining these performance metrics grouped according to the two phases.
 \section{Aggregation Phase}\label{sec:firstphase}
 In this section, we investigate the channel utilization performance and the channel success probability of the aggregation phase.
 \subsection{Resource Scheduling}
 Different from~\cite{Malak-2016} where all the MTDs in the network region are assured to be assigned with an equal channel resource, in our work, aggregators can schedule the limited resources to MTDs. Under the RRS scheme, each channel is allocated to the MTD with the same probability. Under the CRS scheme, we assume that an aggregator with $K$ MTDs has the knowledge of their fading gains. Let $\left\{h_{(1)},...,h_{(i)},...h_{(K)}\right\}$ denote the decreasing ordered fading gains, where $h_{(i-1)}>h_{(i)}$. If $K>N$, the aggregator will pick $N$ \redcom{MTDs} with better channel gains (i.e., the set $\left\{h_{(1)},...h_{(N)}\right\}$) and then assign the channel set $\mathcal{N}$ to the corresponding $N$ MTDs\footnote{Note that we assume the assignment is random. In other words, the MTD with the best channel gain will not always be assigned with the first (or last) channel.}. If $K<N$, the aggregator will randomly pick $K$ channels and allocate these channels to $K$ MTDs. In general, this scheme can guarantee that data from MTDs are more likely to be successfully decoded by the aggregator.
\subsection{Channel Utilization Performance}
As mentioned in Section~\ref{subsec:system:network}, the number of MTDs requiring data transmission $K$ within an aggregator's serving zone is a Poisson random variable, while the number of channels $N$ is fixed. Under the case where the number of MTDs is less than the available resources, a certain channel may not be occupied by a MTD because of the excessive resources compared to the number of MTDs. Similarly, when the number of MTDs is greater than the available resources, a certain MTD may not be assigned a channel due to the insufficient resources.

Using the probability theory, we present the results of $\bar{p}_{O}$ and $\apndrop$ in the following lemmas. Note that these lemmas hold for both the RRS and CRS schemes.
\begin{lemma}
Based on the system model in Section~\ref{sec:system}, the average channel occupation probability is
 \begin{align}\label{eq:occupyprob}
 \bar{p}_O=1-\frac{\Gamma\left[1+N,\bar{m}\right]}{\Gamma\left[N\right]}-\frac{\exp(-\bar{m})\bar{m}^N-\bar{m}\Gamma[N,\bar{m}]}{\Gamma[1+N]},
 \end{align}
 \noindent where $\Gamma[\cdot]$ and $\Gamma[\cdot,\cdot]$ are complete gamma function and the incomplete upper gamma functions, respectively.
\end{lemma}
\begin{lemma}
Based on the system model in Section~\ref{sec:system}, the average non-drop probability of a MTD is
\ifCLASSOPTIONpeerreview
\begin{align}\label{eq:nondropprob}
\apndrop=\frac{\Gamma[1+N,\bar{m}]}{\Gamma[1+N]}+\exp(-\bar{m})\bar{m}^{1+N}N\frac{\,_2F_2\left[\left\{1,1+N\right\},\left\{2+N,2+N\right\},\bar{m}\right]}{(N+1)(N+1)!},
\end{align}
\else
\begin{align}\label{eq:nondropprob}
&\apndrop=\frac{\Gamma[1+N,\bar{m}]}{\Gamma[1+N]}\nonumber\\
&+\exp(-\bar{m})\bar{m}^{1+N}N\frac{\,_2F_2\left[\left\{1,1+N\right\},\left\{2+N,2+N\right\},\bar{m}\right]}{(N+1)(N+1)!},
\end{align}
\fi
\noindent where $_2F_2\left[\left\{\cdot,\cdot\right\},\left\{\cdot,\cdot\right\},\cdot\right]$ is the generalized hypergeometric function.
\end{lemma}

\emph{Proof:} See Appendix A.
\begin{remark}
Both of these two quantities are not impacted by the density of aggregators and BSs, and they rely only on the number of channels $N$ and the average number of MTDs per aggregation $\bar{m}$. We find that increasing $N$ can increase $\apndrop$ (equivalently, less MTDs are likely to be dropped) while it deteriorates the channel occupation performance (i.e., channels are not efficiently utilized). In terms of the impact of $N$, we will examine it in Section~\ref{sec:result:N}.
\end{remark}

 \subsection{Average Channel Success Probability}
According to Section~\ref{subsec:system:transmission}, the channel set $\mathcal{N}$ is used among all aggregators. For a typical link, it will receive the inter-cluster interference from other MTDs occupying the same channel. Based on the \redcom{independent thinning property and} displacement theorem in stochastic geometry~\cite{Haenggi-2012}, the location of the interfering MTDs that generates the interference on a \redcom{typical} channel for a typical aggregator is in fact the HPPP with density $\bar{p}_O\adensity$, denoted as $\Phi_{\textrm{MTD}}^{\redcom{\textrm{served}}}$.

 To calculate the inter-cluster interference, we condition on having a typical aggregator located at the origin. By Slivnyak's theorem~\cite{Haenggi-2012}, conditioning a node at a certain location does not change the distribution of the rest of the process. Using the stochastic geometry, we obtain the aggregation phase's channel success probability under the RRS as follows.
\begin{theorem}\label{theo:channeloutage1}
Based on the system model in Section~\ref{sec:system}, under the random resource scheduling scheme, the aggregation phase's average channel success probability experienced at a typical aggregator for a certain channel is
\begin{align}\label{eq:result_channelRRS}
\apouta^{r}=\sum\limits_{t=0}\limits^{m_1-1}\frac{(-s)^{t}}{t!}\left.\frac{d^t}{d s^t}\mathcal{M}_{I_{1}}(s)\right|_{s=m_1\gamma_1},
\end{align}
\noindent where $\mathcal{M}_{I_{1}}(s)=\exp\left(-\bar{p}_O\adensity\pi\frac{R_s^2}{2}\Gamma\!\left[1+\frac{2}{\alpha}\right]\Gamma\!\left[1-\frac{2}{\alpha}\right]s^{\frac{2}{\alpha}}\right)$ is the moment generating function (MGF) of the inter-cluster interference $I_1$ and $m_1$ is the fading parameter for the desired link.
\end{theorem}

\textit{Proof:} See Appendix B.
\redcom{\begin{remark}
From~\eqref{eq:result_channelRRS}, we can see that $\apouta^{r}$ is a decreasing function of $\bar{p}_O$, $\adensity$ and $R_s$, since these terms appear in the exponent in $\mathcal{M}_{I_{1}}(s)$. This implies that increasing the density of aggregators $\adensity$ or the radius of the serving zone $R_s$ or the average channel occupation probability $\bar{p}_O$ (equivalently either increasing $\bar{m}$ or decreasing $N$) can degrade the average channel success probability of the aggregation phase. Note that~\eqref{eq:result_channelRRS} is a well-known result in the stochastic geometry literature when the location of interfering nodes follows a HPPP.
\end{remark}}

Compared to the RRS, the analysis of the channel success probability for the CRS is more complex because its performance relies on the number of MTDs requiring data transmission. The result is presented in the following proposition.
\ifCLASSOPTIONpeerreview
 \begin{proposition}\label{theo:channeloutage1p}
 Based on the system model in Section~\ref{sec:system}, under the channel-aware resource scheduling scheme, the aggregation phase's average channel success probability experienced at a typical aggregator for a certain channel is
 \begin{align}
 \apouta^{c}=&\apouta^{r}\frac{\exp(\bar{m})\Gamma\![1+N,\bar{m}]-\Gamma\![1+N]}{\left(\exp(\bar{m})-1\right)\Gamma\![1+N]}+\sum_{k=N+1}^{\infty}\pouta^{c}(k)\frac{\Pr(K=k)}{1-\Pr(K=0)}, \label{eq:conprob1pppp}\\
\pouta^{c}(K)\!\approx&\!\frac{\sum\limits_{i=1}^{N}\!\left(\frac{1}{2}\!-\!\frac{1}{\pi}\int_0^{\infty}\!\!\frac{1}{w}\textrm{Im}\!\left\{\!\exp\!\left(\!-\bar{p}_O\adensity\pi\frac{R_s^2}{2}\Gamma\!\left[\frac{\alpha\!+\!2}{\alpha}\right]\Gamma\!\left[\frac{\alpha\!-\!2}{\alpha}\right]
 \!\left(-\ii w\right)^{\!\frac{2}{\alpha}}\right)\!\exp\!\left(\!-\ii w\frac{\mathbb{E}_{h_{(i),k}}\left\{h_{(i),k}\right\}}{\gamma_1}\!\right)\!\!\right\}\textup{d}w\!\!\right)}{N},\label{eq:conprob1p}
 \end{align}
 \noindent where $\pouta^{c}(k)$ is the conditional channel success probability of the aggregation phase which is conditioned on the number of MTDs requiring data transmission $k$, $\apouta^{r}$ is presented in Theorem~\ref{theo:channeloutage1}, $h_{(i),k}$ is the $i$-th best fading gain given $k$ MTDs within the typical aggregator's serving zone and $\mathbb{E}_{h_{(i),k}}\!\left\{h_{(i),k}\right\}$ is its corresponding mean which can be obtained using the formula in~\cite[eq.(2.2)]{Gupta-1960} or numerical evaluation via Mathematica.
 \end{proposition}
\else
 \begin{proposition}\label{theo:channeloutage1p}
 Based on the system model in Section~\ref{sec:system}, under the channel-aware resource scheduling scheme, the aggregation phase's average channel success probability experienced at a typical aggregator for a certain channel is shown in~\eqref{eq:conprob1pppp} and~\eqref{eq:conprob1p}, at the top of next page, where $\pouta^{c}(k)$ is the conditional channel success probability of the aggregation phase which is conditioned on the number of MTDs requiring data transmission $k$, $\apouta^{r}$ is presented in Theorem~\ref{theo:channeloutage1}, $h_{(i),k}$ is the $i$-th best fading gain given $k$ MTDs within the typical aggregator's serving zone and $\mathbb{E}_{h_{(i),k}}\!\left\{h_{(i),k}\right\}$ is its corresponding mean which can be obtained using the formula in~\cite[eq.(2.2)]{Gupta-1960} or numerical evaluation via Mathematica.
   \begin{figure*}[!t]
\normalsize
 \begin{align}
 \apouta^{c}=&\apouta^{r}\frac{\exp(\bar{m})\Gamma\![1+N,\bar{m}]-\Gamma\![1+N]}{\left(\exp(\bar{m})-1\right)\Gamma\![1+N]}+\sum_{k=N+1}^{\infty}\pouta^{c}(k)\frac{\Pr(K=k)}{1-\Pr(K=0)}, \label{eq:conprob1pppp}\\
\pouta^{c}(K)\!\approx&\!\frac{\sum\limits_{i=1}^{N}\!\left(\frac{1}{2}\!-\!\frac{1}{\pi}\int_0^{\infty}\!\!\frac{1}{w}\textrm{Im}\!\left\{\!\exp\!\left(\!-\bar{p}_O\adensity\pi\frac{R_s^2}{2}\Gamma\!\left[\frac{\alpha\!+\!2}{\alpha}\right]\Gamma\!\left[\frac{\alpha\!-\!2}{\alpha}\right]
 \!\left(-\ii w\right)^{\!\frac{2}{\alpha}}\right)\!\exp\!\left(\!-\ii w\frac{\mathbb{E}_{h_{(i),k}}\left\{h_{(i),k}\right\}}{\gamma_1}\!\right)\!\!\right\}\textup{d}w\!\!\right)}{N},\label{eq:conprob1p}
 \end{align}
 \hrulefill
\vspace*{4pt}
\vspace{-0.05 in}
\end{figure*}
 \end{proposition}
\fi

 \textit{Proof:} See Appendix C.
  \begin{figure}
        \centering
        \includegraphics[width=0.5  \textwidth]{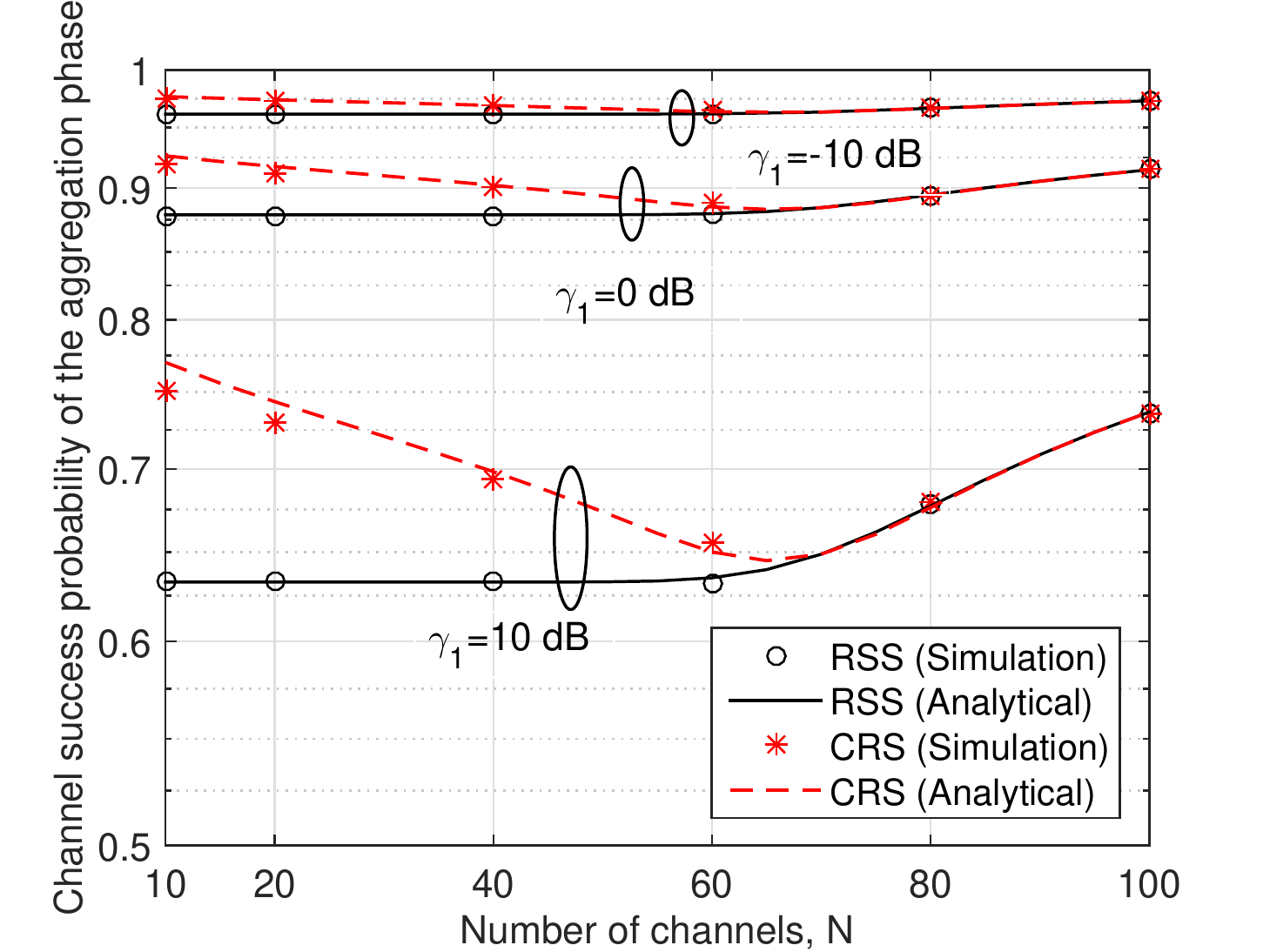}
        \vspace{-1.7em}\caption{Channel success probability of the aggregation phase, $\apouta$, versus the number of channels, $N$, with \redcom{$\bar{m}=70$} and different SIR threshold for both of the RRS and CRS.}
        \label{fig_channeloutage1}
\end{figure}
\redcom{\begin{remark}
To the best of our knowledge, $\apouta^{c}$ in Proposition~\ref{theo:channeloutage1p} is a new result in the literature. From~\eqref{eq:conprob1pppp}, we can see that when $N$ is very large, especially larger than $\bar{m}$, the term $\frac{\Pr(K=k)}{1-\Pr(K=0)}$ can become negligibly small such that $\apouta^{c}$ is almost the same as $\apouta^{r}$, i.e., when $N$ is very large, the CRS and RRS schemes perform the same in terms of the aggregation phase's average channel success probability.
\end{remark}}

\textit{\underline{Validation and Insights:}} Before ending this section, we validate the analytical results of the channel success probability for two schemes. The adopted system parameters are listed in Table~\ref{tb:para}. Fig.~\ref{fig_channeloutage1} shows that the analytical results provide a good match with the simulation results, and the small gap for the CRS comes from the Jensen's inequality for Proposition~\ref{theo:channeloutage1p}. As illustrated in Fig.~\ref{fig_channeloutage1}, when $N$ is small, $\apouta^{c}$ for CRS outperforms $\apouta^{r}$ for RRS. As $N$ increases, the performance difference between RRS and CRS becomes less obvious and finally their curves overlap. This is due to the fact that when $N$ is large (i.e., comparable to or larger than $\bar{m}$), most of the time the number of MTDs is less than the available resources such that the implementation of the CRS is almost the same as the RRS. \redcom{This is in agreement with our observation in Remark 3.}

Additionally, as the number of channels increases, the success probability for RRS keeps constant at first and then increases. This is because the larger value of $N$ leads to the lower occupation probability for a certain channel and the interference on the channel is reduced which improves the success probability. Unlike the RRS, the channel success probability of CRS decreases at first and then increases. When $N$ is really small, after selecting the better fading, the fading gain for CRS is far greater than the fading gain without ordering and consequently the channel
success probability is much better. Once $N$ increases, the fading scenario tends to be the case without ordering.

 \section{Relaying Phase}\label{sec:secondphase}
 \subsection{Number of Active Channels}
 The average channel success probability derived in Theorem~\ref{theo:channeloutage1} and Proposition~\ref{theo:channeloutage1p} is valid for one channel only. There are a total number of $N$ channels for each aggregator and it is necessary to investigate the number of active channels (i.e., the \redcom{data} from \redcom{served} MTDs can be successfully decoded by the aggregator on these channels) based on the derived channel success probability of the aggregation phase. Such a distribution is relevant to the metric~\textit{average number of successful MTDs} from~\eqref{eq:generalnumbersuc} and also plays a key role in determining the channel success probability of the relaying phase as shown in the following subsection. The results of this distribution for the considered resource scheduling schemes are presented below.

\begin{corollary}\label{coro:distribution1}
Based on the system model in Section~\ref{sec:system}, under the random resource scheduling scheme, the PMF of the number of active channels, $K_1$, for a typical aggregator is approximated by
\ifCLASSOPTIONpeerreview
\begin{align}\label{eq:distribution1}
\Pr\!\,^{r}\!\left(K_1=k_1\right)\approx&\frac{\exp\!\left(-\bar{m}\apouta^{r}\right)\bar{m}^{1+N}(\apouta^{r})^{k_1}\left(1\!-\!\apouta^{r}\right)^{1-k_1+N}E_{k_1-N}[\bar{m}(1-\apouta^{r})]}{\Gamma\![1+k_1]\Gamma\![1-k_1+N]}\nonumber\\
&+\frac{\binom{N}{k_1}(\apouta^{r})^{k_1}\left(1-\apouta^{r}\right)^{N-k_1}\Gamma\![1+N,0,\bar{m}]}{\Gamma\![1+N]},
\end{align}
\else
\begin{align}\label{eq:distribution1}
&\Pr\!\,^{r}\!\left(K_1=k_1\right)\approx\frac{\binom{N}{k_1}\Gamma\![1+N,0,\bar{m}]}{(\apouta^{r})^{-k_1}\left(1-\apouta^{r}\right)^{k_1-N}\Gamma\![1+N]} \nonumber\\
&+ \frac{\exp\!\left(-\bar{m}\apouta^{r}\right)\bar{m}^{1+N}E_{k_1-N}[\bar{m}(1-\apouta^{r})]}{(\apouta^{r})^{-k_1}\left(1\!-\!\apouta^{r}\right)^{k_1-N-1}\Gamma\![1+k_1]\Gamma\![1-k_1+N]},
\end{align}
\fi
\noindent where $E_n[z]$ denotes the exponential integral function, $\apouta^{r}$ is presented in Theorem~\ref{theo:channeloutage1} and $k_1$ is in the range of $[0,N]$.
\end{corollary}

\textit{Proof:} See Appendix D.

\begin{corollary}\label{coro:distribution1p}
Based on the system model in Section~\ref{sec:system}, under the channel-aware resource scheduling scheme, the PMF of the number of active channels, $K_1$, for a typical aggregator is approximated by
\ifCLASSOPTIONpeerreview
\begin{align}\label{eq:distribution2}
\Pr\!\,^{c}\!\left(K_1=k_1\right)\approx&\frac{\exp\!\left(-\bar{m}\apouta^{r}\right)\bar{m}^{1+N}(\apouta^{r})^{k_1}\left(1\!-\!\apouta^{r}\right)^{1-k_1+N}E_{k_1-N}[\bar{m}(1-\apouta^{r})]}{\Gamma\![1+k_1]\Gamma\![1-k_1+N]}\nonumber\\
&+\sum_{k=N+1}^{\infty}\binom{N}{k_1}\left(\pouta^{c}(k)\right)^{k_1}(1-\pouta^{c}(k))^{N-k_1}\frac{\bar{m}^k\exp(-\bar{m})}{k!},
\end{align}
\else
\begin{align}\label{eq:distribution2}
&\Pr\!\,^{c}\!\left(K_1=k_1\right)\nonumber\\
&\approx\!\!\!\sum_{k=N+1}^{\infty}\!\binom{N}{k_1}\left(\pouta^{c}(k)\right)^{k_1}(1-\pouta^{c}(k))^{N-k_1}\frac{\bar{m}^k\exp(-\bar{m})}{k!}\nonumber\\
&+
\frac{\exp\!\left(-\bar{m}\apouta^{r}\right)\bar{m}^{1+N}E_{k_1-N}[\bar{m}(1-\apouta^{r})]}{(\apouta^{r})^{-k_1}\left(1\!-\!\apouta^{r}\right)^{k_1-N-1}\Gamma\![1+k_1]\Gamma\![1-k_1+N]},
\end{align}
\fi
\noindent where $\pouta^{c}(k)$ is presented in~\eqref{eq:conprob1p}.
\end{corollary}
\begin{proof}
The proof is similar to the derivation of Corollary~\ref{coro:distribution1}. But in~\eqref{eq:distK1}, when $K>N$, $\Pr(K_1=k_1|K)$ becomes $\binom{N}{k_1}\left(\pouta^{c}(K)\right)^{k_1}(1-\pouta^{c}(K))^{N-k_1}$ since the channel success probability depends on the number of MTDs requiring data transmission.
\end{proof}
  \subsection{Average Channel Success Probability}
 In the relaying phase, each aggregator transmits its aggregated data to its nearest BS in a round-robin fashion. We assume that the aggregator with at least one active channel (i.e., $K_1\geq 1$) can do the data transmission and it is called the active aggregator. Otherwise, the aggregator (i.e., $K_1=0$) becomes silent. Note that the location of active aggregators is somewhat correlated. \redcom{For example, for those aggregators that are very close to each other, the interference on each channel may be very severe. Hence, it is more likely that these aggregators become inactive.} For analytical tractability, we model the location of these active aggregators as a HPPP with density $\adensity'=(1-\Pr(K_1=0))\adensity$.\footnote{In this work, $\Pr(K_1=k_1)$ is the general distribution of the number of active channels. $\Pr(K_1=k_1)=\Pr\!\,^{r}(K_1=k_1)$ for the RRS, and $\Pr(K_1=k_1)=\Pr\!\,^{c}(K_1=k_1)$ for the CRS.}

 For an active aggregator with $K_1$ active channels, its aggregated data can be successfully decoded by its associated BS as long as its SIR meets the following condition~\cite{Singh-2015,7136376}
 \begin{align}\label{eq:outage2}
 D K_1\leq \frac{T W}{N_a}\log\left(1+\textsf{SIR}_2\right),
 \end{align}
 \noindent where $D$ is the data size for each MTD, $T$ is the total transmission time for the relaying phase, $W$ is the available bandwidth for each BS and $N_a$ is the total number of active aggregators associated to the same BS.

 By re-arranging~\eqref{eq:outage2}, we can then write the average channel success probability for the relaying phase as
 \begin{align}
\apoutb=\mathbb{E}_{K_1,N_a}\left\{\Pr\left(\textsf{SIR}_2\geq2^{\frac{D K_1 N_a}{TW}}-1\right)\right\}.
 \end{align}

 \noindent Here we define $\gamma_2
\triangleq 2^{\frac{D K_1 N_a}{TW}}-1$, which depends on the number of active channel $K_1$ and the number of associated active aggregators $N_a$. The main result of the channel success probability of the relaying phase is presented in the following.

\begin{proposition}\label{theo:channeloutage2}
Based on the system model described in Section~\ref{sec:system}, the average channel success probability experienced at a typical BS from one of its associated active MTD is
\ifCLASSOPTIONpeerreview
\begin{align}
\apoutb=&\sum_{k_1=1}^{N}\poutb(k_1)\frac{\Pr(K_1=k_1)}{1-\Pr(K_1=0)},\label{eq:outage2_final}\\
\poutb(k_1)=&\sum_{n_a=1}^{\infty}\left(\sum\limits_{t=0}\limits^{m_2-1}\frac{(-s)^{t}}{t!}\left.\frac{d^t}{d s^t}\mathcal{M}_{I_{2}}(s)\right|_{s=m_2\left(2^{\frac{D k_1 n_a}{W}}-1\right)}\right)\frac{\Pr\left(N_a=n_a\right)}{1-\Pr\left(N_a=0\right)},\label{eq:outage2_con}
\end{align}
\else
\begin{align}
\apoutb=&\sum_{k_1=1}^{N}\poutb(k_1)\frac{\Pr(K_1=k_1)}{1-\Pr(K_1=0)},\label{eq:outage2_final}\\
\poutb(k_1)=&\!\sum_{n_a=1}^{\infty}\!\left(\sum\limits_{t=0}\limits^{m_2-1}\!\frac{(-s)^{t}}{t!}\left.\frac{d^t}{d s^t}\mathcal{M}_{I_{2}}(s)\right|_{s=m_2\left(2^{\frac{D k_1 n_a}{W}}\!-\!1\right)}\!\right)\nonumber\\
&\times\frac{\Pr\left(N_a=n_a\right)}{1-\Pr\left(N_a=0\right)},\label{eq:outage2_con}
\end{align}
\fi
\noindent where $\mathcal{M}_{I_{2}}(s)=\exp\left(-2(1-p_{\textrm{void}})s\frac{\,_2F_1\left[1,1-\frac{2}{\alpha},2-\frac{2}{\alpha},-s\right]}{\alpha-2}\right)$, $p_{\textrm{void}}=\left(1+\frac{\adensity'}{3.5\BSdensity}\right)^{-3.5}$ is the void probability, $\poutb(k_1)$ is the conditional channel success probability of the relaying phase which is conditioned on the number of active channels for an aggregator, $\Pr(K_1=k_1)$ is given in Corollaries~\ref{coro:distribution1} and~\ref{coro:distribution1p} for the two different schemes, and $\Pr\left(N_a=n_a\right)$ is the PMF of the number of active aggregators associated to the typical BS which is given in~\eqref{eq:Nadist}.
\end{proposition}

\textit{Proof:} See Appendix E.
\redcom{\begin{remark}
 For the special case of $m_2=1$, which corresponds to the case that the typical link experiences Rayleigh fading, we can have $\apoutb\approx\exp\left(-2(1-p_{\textrm{void}})s\frac{\,_2F_1\left[1,1-\frac{2}{\alpha},2-\frac{2}{\alpha},-s\right]}{\alpha-2}\right)$, where
$s\approx 2^{\frac{D \mathbb{E}_{K_1}[K_1]\mathbb{E}_{N_a}[N_a]}{W}}-1$, $\mathbb{E}_{K_1}[K_1]\approx N \bar{p}_O \apouta$ and $\mathbb{E}_{N_a}[N_a]=\frac{\adensity'}{\BSdensity}$. This shows that $\apoutb$ is a monotonic decreasing function of $s$. Thus if the density of BSs $\BSdensity$ increases, both $s$ and $1-p_{\textrm{void}}$ decrease, which improves $\apoutb$. If the number of channels $N$ increases, it increases $s$ as well as $\mathbb{E}_{K_1}[K_1]$ and hence $\apoutb$ degrades. Note that $\poutb(k_1)$ in~\eqref{eq:outage2_con} has been studied in the stochastic geometry literature with different system setups, e.g.,~\cite{Singh-2015,7136376,6516885}. However, in this paper, we are interested in $\apoutb$ in~\eqref{eq:outage2_final}.
\end{remark}}
  \begin{figure}
        \centering
        \includegraphics[width=0.5  \textwidth]{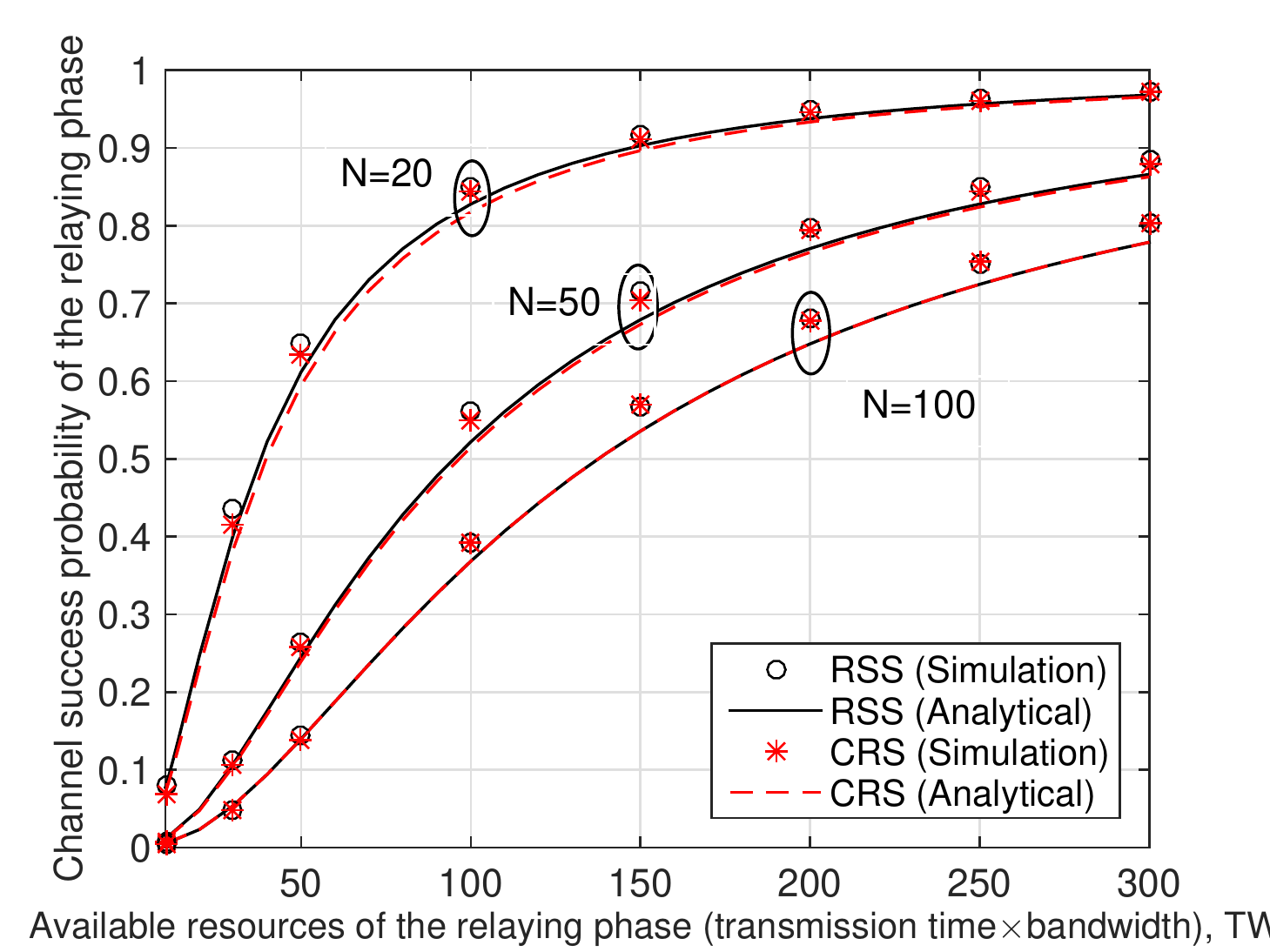}
        \vspace{-1.7em}\caption{Channel success probability of the relaying phase, $\apoutb$, versus the available resources, \redcom{transmission time$\times$bandwidth,} $TW$, for both RRS and CRS.}
        \label{fig_channeloutage2}
\end{figure}

\textit{\underline{Validation and Insights:}} Before ending this subsection, we validate the channel success probability of the relaying phase by comparing the analytical results with the simulation results, since the derived channel success probability is an approximate result. According to Fig.~\ref{fig_channeloutage2}, we can see that the derived analytical results match the simulation results fairly well. By comparing the solid lines with the dashed lines, we observe that when the number of channels $N$ is small, the channel success probability under CRS performs slightly worse than the RRS. Once $N$ gets larger, the channel success probability under two scheme is the same. This is due to the fact that the channel success probability is determined by the $\gamma_2$ (equivalently, the channel success probability of the aggregation phase). From Fig.~\ref{fig_channeloutage1}, we already find that $\apouta^{c}$ outperforms $\apouta^{r}$ when $N$ is small. In other words, under this scenario, the CRS is more likely to allow more MTDs to be successfully served by aggregators thereby increasing the probability of being a large value for $\gamma_2$. Consequently, it deteriorates the channel success probability of the relaying phase. When $N$ is large, there is no performance difference between the CRS and RRS for the aggregation phase, which leads to the same performance for the relaying phase.

Moreover, Fig.~\ref{fig_channeloutage2} implies that increasing the number of channels can worsen the channel success probability of the relaying phase. This is because a larger value of $N$ allows the aggregators to aggregate much more data from MTDs, which expands the distribution of active channels. As a result, the channel success probability of the relaying phase is reduced. However, Fig.~\ref{fig_channeloutage1} shows that increasing $N$ can improve the performance of the aggregation phase. The impact of $N$ on the overall performance will be assessed in Section~\ref{sec:result:N}. Before that, we present a summary of the key performance metrics in the next section.
\section{Summary of Key Performance Metrics}\label{sec:summary:metric}
With the facilitation of the previous derived results, we now can formulate the key performance metrics as summarized in Table~\ref{tb:2}.

Note that for the CRS scheme, the MTD success probability and probability of successful channel utilization are different from~\eqref{eq:summary_channel1} and \eqref{eq:summary_ulti}. We find that the non-drop probability (and occupy probability) and the aggregation phase's channel success probability are dependent as they are related to the number of MTDs $K$. Thus, we have
\ifCLASSOPTIONpeerreview
\begin{align} \label{eq:summary_channel2}
\psuc^{c}\approx&\left(\sum_{k=0}^{N}1\frac{\bar{m}^k\exp(-\bar{m})}{k!}\apouta^{r}+\sum_{k=N+1}^{\infty}\frac{N}{k}\pouta^{c}(k)\frac{\bar{m}^k\exp(-\bar{m})}{k!}\right)\apoutb^{c} \nonumber\\
=&\left(\frac{\Gamma\![1+N,\bar{m}]}{\Gamma\![1+N]}\apouta^{r}+\sum_{k=N+1}^{\infty}\frac{N}{k}\pouta^{c}(k)\frac{\bar{m}^k\exp(-\bar{m})}{k!}\right)\apoutb^{c},
\end{align}
\begin{align}\label{eq:summary_ulti2}
\achannel^{c}\approx &\left(\frac{\bar{m}\Gamma[1+N,\bar{m}]-\exp(-\bar{m})\bar{m}^{N+1}}{N^2\Gamma[N]}\apouta^{r}+\sum_{k=N+1}^{\infty}\pouta^{c}(k)\frac{\bar{m}^k\exp(-\bar{m})}{k!}\right)\apoutb^{c}.
\end{align}
\else
\begin{align} \label{eq:summary_channel2}
&\psuc^{c}\approx\left(\sum_{k=0}^{N}1\frac{\bar{m}^k\exp(-\bar{m})}{k!}\apouta^{r}\right.\nonumber\\
&\left.\quad\quad\quad\quad+\sum_{k=N+1}^{\infty}\frac{N}{k}\pouta^{c}(k)\frac{\bar{m}^k\exp(-\bar{m})}{k!}\right)\apoutb^{c} \nonumber\\
=&\left(\!\frac{\Gamma\![1+N,\bar{m}]}{\Gamma\![1+N]}\apouta^{r}+\!\!\sum_{k=N+1}^{\infty}\!\!\frac{N}{k}\pouta^{c}(k)\frac{\bar{m}^k\exp(-\bar{m})}{k!}\!\right)\apoutb^{c},
\end{align}
\begin{align}\label{eq:summary_ulti2}
\achannel^{c}\approx &\left(\frac{\bar{m}\Gamma[1+N,\bar{m}]-\exp(-\bar{m})\bar{m}^{N+1}}{N^2\Gamma[N]}\apouta^{r}\right.\nonumber\\
&\left.+\sum_{k=N+1}^{\infty}\pouta^{c}(k)\frac{\bar{m}^k\exp(-\bar{m})}{k!}\right)\apoutb^{c}.
\end{align}
\fi

Moreover, when we analyze the effect of aggregator's density in Section~\ref{sec:effect:dense}, we also consider the metric, average number of successful MTDs per km$^2$, which is related to average number of successful MTDs by $\lambda_a\bar{K}_{\textrm{suc}}$.
\begin{table*}[t]
\centering
\caption{Summary of the Key Performance Metric.}\label{tb:2}
\begin{tabular}{|c||c|c|c|c|c|c|c|} \hline
\multirow{2}{*}{scheme}&\multirow{2}{*}{metric}&\multirow{2}{*}{\makecell[c]{general\\form}}& \multirow{2}{*}{$\apndrop$}& \multirow{2}{*}{$\bar{p}_{O}$}& \multirow{2}{*}{\makecell[c]{$\apouta$,\\$\pouta(k)$}}& \multirow{2}{*}{{\makecell[c]{$\apoutb$,\\$\poutb(k_1)$}}}&\multirow{2}{*}{$\Pr(K_1=k_1)$} \\
&&&&&&&\\\cline{2-8}\hline
\multirow{3}{*}{\begin{sideways}RRS\end{sideways}}
&MTD success probability & \eqref{eq:summary_channel1} & \eqref{eq:nondropprob}&  & \eqref{eq:result_channelRRS}& \eqref{eq:outage2_final}$\&$\eqref{eq:outage2_con}  &\eqref{eq:distribution1}\\
&average number of successful MTDs &\eqref{eq:generalnumbersuc}  & &  & &$\eqref{eq:outage2_con}$ &\eqref{eq:distribution1}\\
&probability of successful channel utilization & \eqref{eq:summary_ulti} & &\eqref{eq:occupyprob}  & \eqref{eq:result_channelRRS}&\eqref{eq:outage2_final}$\&$\eqref{eq:outage2_con} &\eqref{eq:distribution1}\\\hline
\multirow{3}{*}{\begin{sideways}CRS\end{sideways}}
&MTD success probability& \eqref{eq:summary_channel2} & &  & \eqref{eq:conprob1p}& \eqref{eq:outage2_final}$\&$\eqref{eq:outage2_con}&\eqref{eq:distribution2}\\
&average number of successful MTDs &\eqref{eq:generalnumbersuc}  & &  & &$\eqref{eq:outage2_con}$ &\eqref{eq:distribution2}\\
&probability of successful channel utilization &\eqref{eq:summary_ulti2}  & &  &\eqref{eq:conprob1p} &\eqref{eq:outage2_final}$\&$\eqref{eq:outage2_con} &\eqref{eq:distribution2}\\\hline
\end{tabular}
\end{table*}

 \section{Numerical Results}\label{sec:result}
 \ifCLASSOPTIONpeerreview
\else
\begin{figure}[h!]
\centering
\subfigure[MTD success probability.]{\label{metric1}\includegraphics[width=0.41\textwidth]{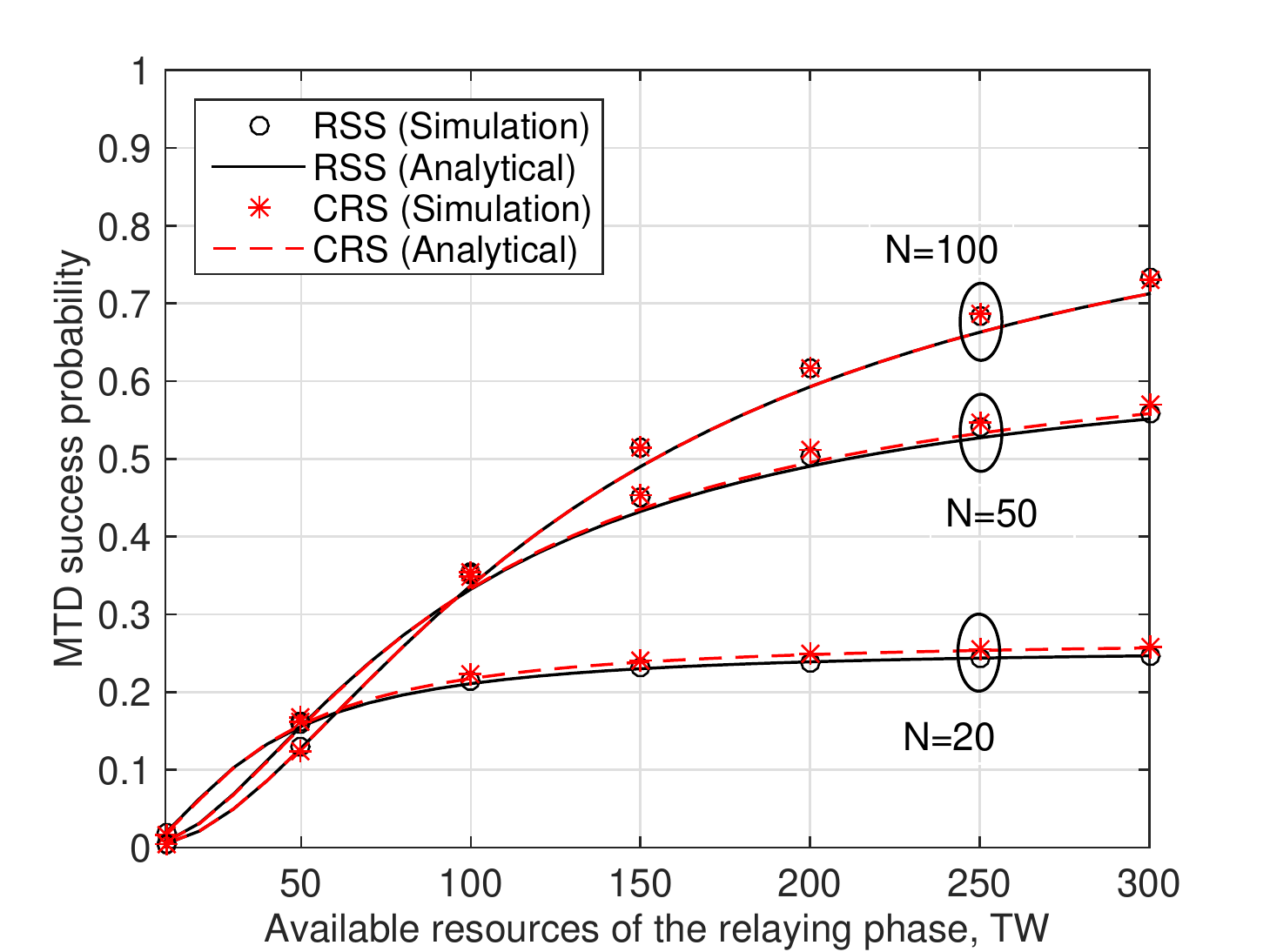}}
\subfigure[Average number of successful MTDs.]{\label{metric2}\includegraphics[width=0.41\textwidth]{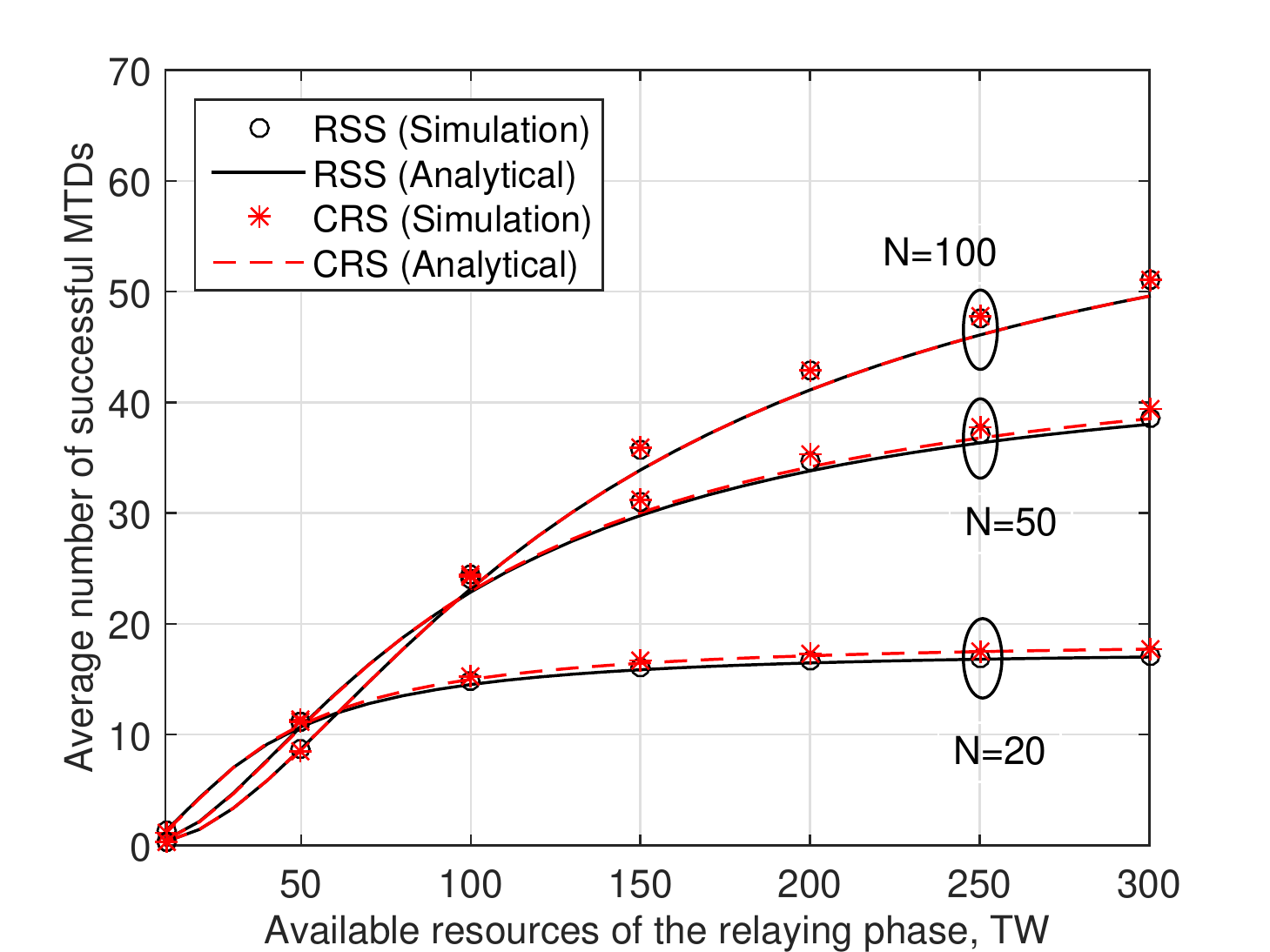}}
\subfigure[Probability of successful channel utilization.]{ \label{metric3}\includegraphics[width=0.41\textwidth]{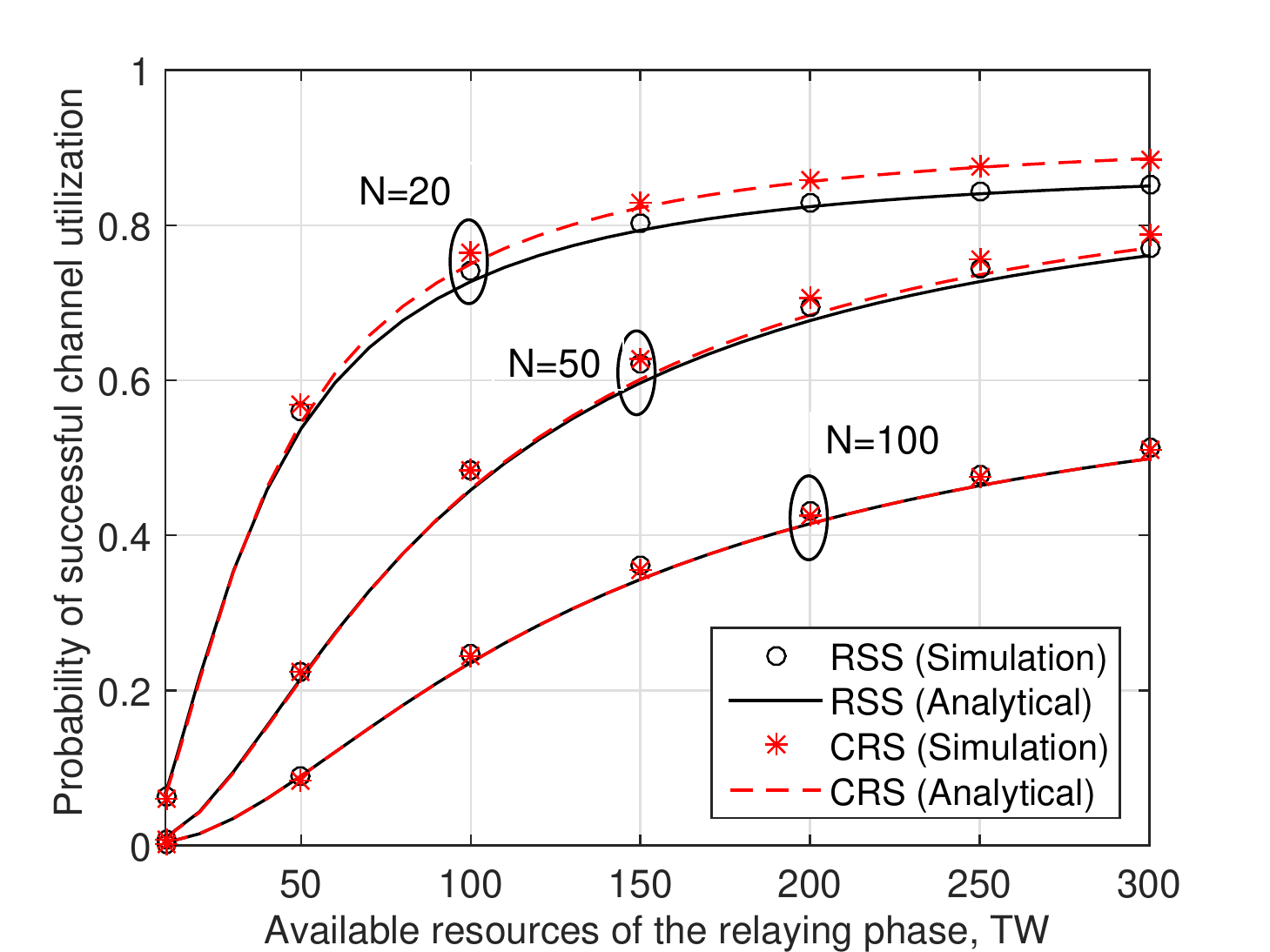}}
\caption{Available resources of the relaying phase, $TW$, versus the considered metrics with different number of channels $N$.}\label{fig_valid}
\end{figure}
\fi
  \begin{table}
\centering
\caption{Main System Parameter Values.}
\label{tb:para}
\begin{tabular}{|l|l|l|}\hline
Parameter & Symbol & Value \\
\hline
Density of BS  & $\BSdensity$ & $\frac{1}{\pi 500^2}$ $/$ m$^2$ (1.27 $/$ km$^2$)\\ \hline
Density of aggregators  & $\adensity$ & $10^{-4.5}$ $/$ m$^2$ (31.667 $/$ km$^2$)\\ \hline
Radius of serving zone & $R_s$ & 50 m \\ \hline
Path-loss exponent & $\alpha$ & $4$ \\ \hline
Average number of MTDs per aggregator & $\bar{m}$ & 70 \\ \hline
SIR threshold of the aggregation phase & $\gamma_1$ &$0$ dB \\ \hline
Nakagami-$m$ fading parameters & $m_1, m_2$ & $m_1=4, m_2=2$ \\ \hline
\end{tabular}
\end{table}

 In this section, we present the numerical results to investigate the performance of machine type communication with data aggregation. \redcom{Note that for the CRS scheme,~\eqref{eq:conprob1pppp},~\eqref{eq:distribution2},~\eqref{eq:outage2_final},~\eqref{eq:summary_channel2} and~\eqref{eq:summary_ulti2} involve infinite summations. We find that, for most cases, the summations converge after about 120 terms and we use this value to evaluate them.} To validate the numerical results, we also present simulation results which are generated using Matlab over $50000$ Monte Carlo simulation runs. \redcom{In order to eliminate the impact of boundary effects in the simulations, the BSs and aggregators are distributed in a disk with radius $3$ km and $6$ km, respectively.} Unless specified otherwise, the values of the main system parameters shown in Table~\ref{tb:para} are used. \redcom{Note that in Table~\ref{tb:para} we set different values for $m_1$ and $m_2$. This is because the aggregator's serving zone is much smaller compared to the BS's coverage region. Thus the link between a MTD and its serving aggregator is more likely to experience less severe fading than the link between an aggregator and its associated BS.} 

\subsection{Validation and Comparison of Schemes}
\ifCLASSOPTIONpeerreview
\begin{figure}[t]

\centering
\subfigure[MTD success probability.]{\label{metric1}\includegraphics[width=0.44\textwidth]{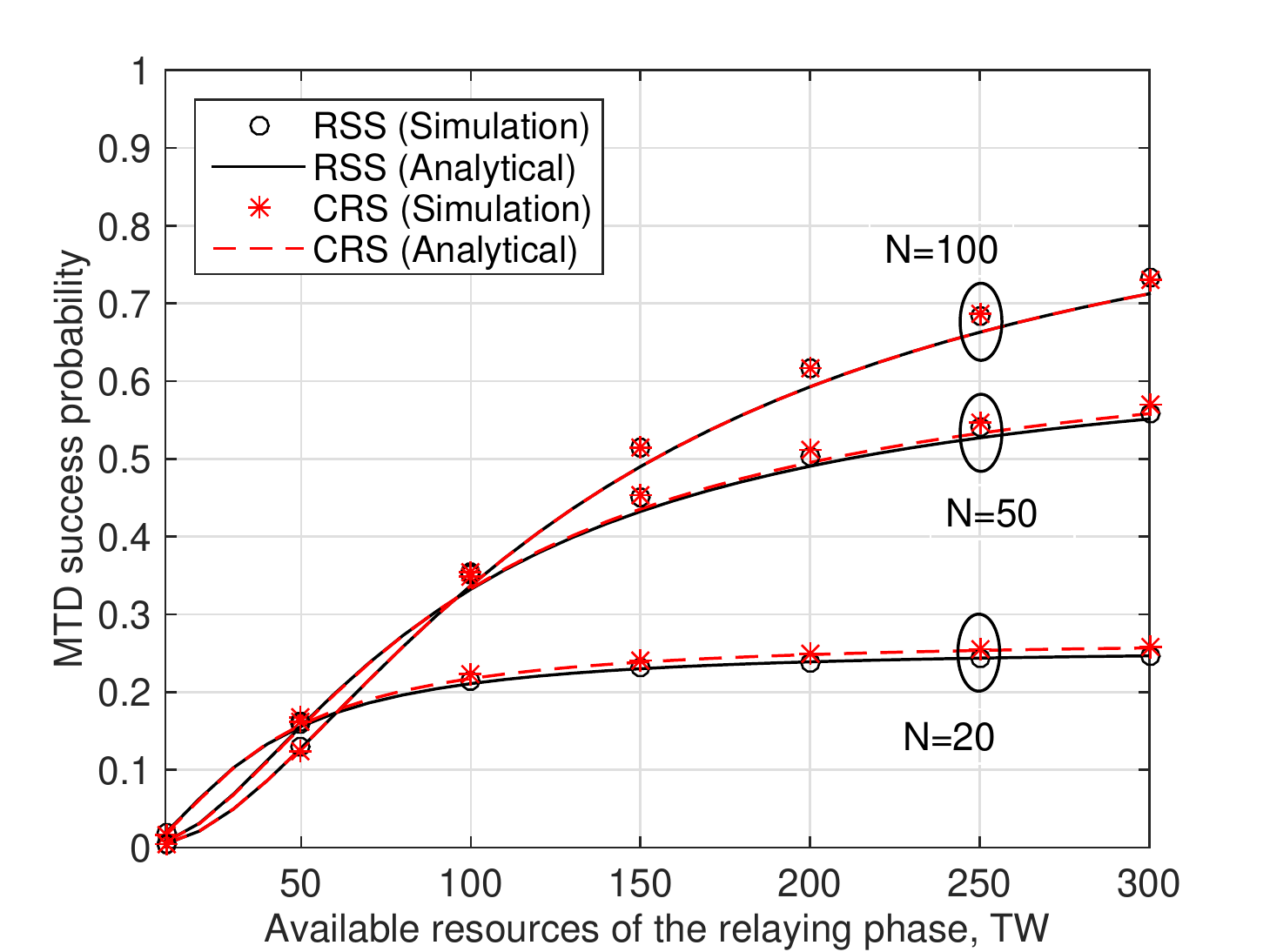}}
\\
\subfigure[Average number of successful MTDs.]{\label{metric2}\includegraphics[width=0.44\textwidth]{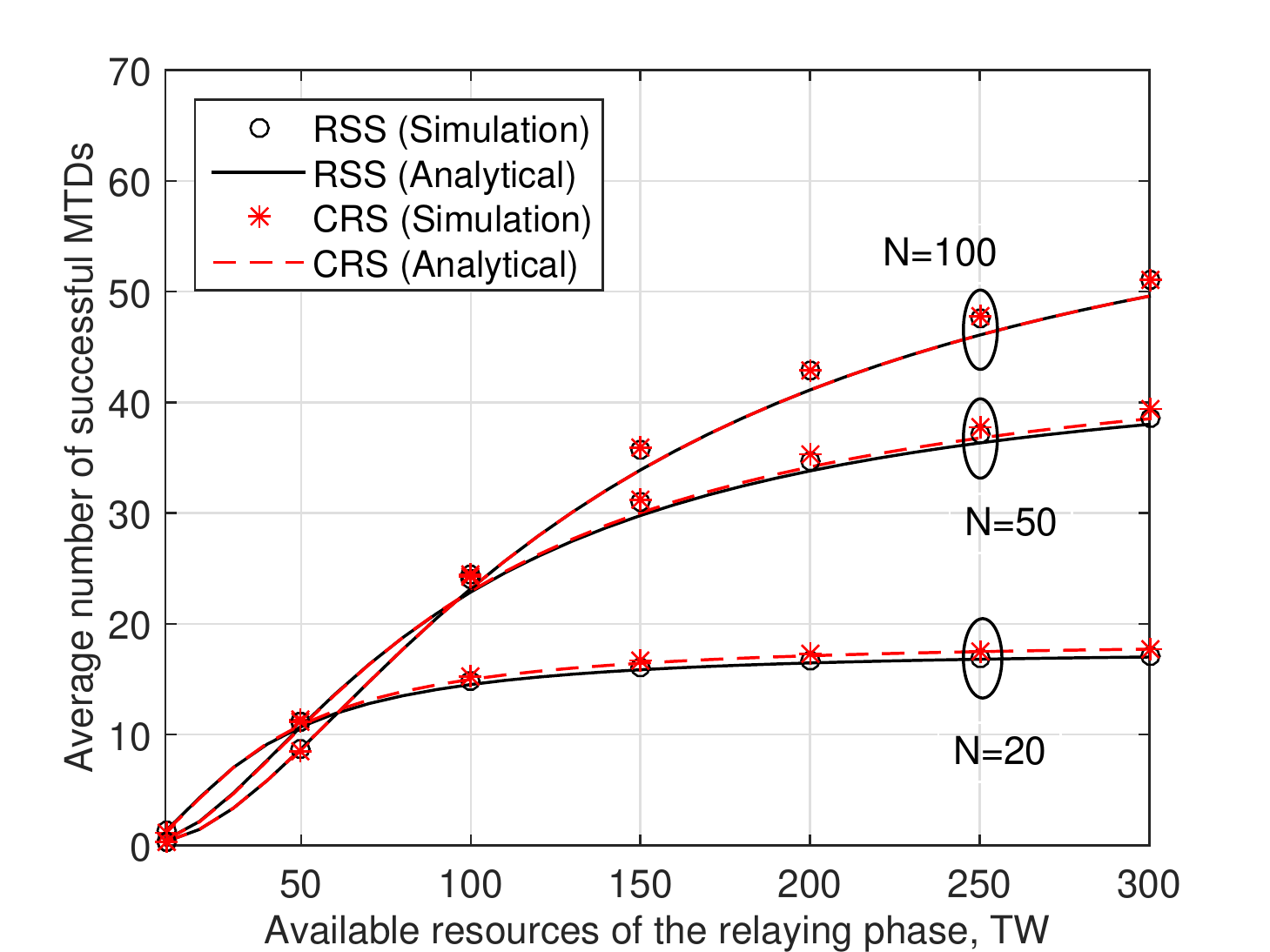}}
\mbox{\hspace{0.5cm}}
\subfigure[Probability of successful channel utilization.]{ \label{metric3}\includegraphics[width=0.44\textwidth]{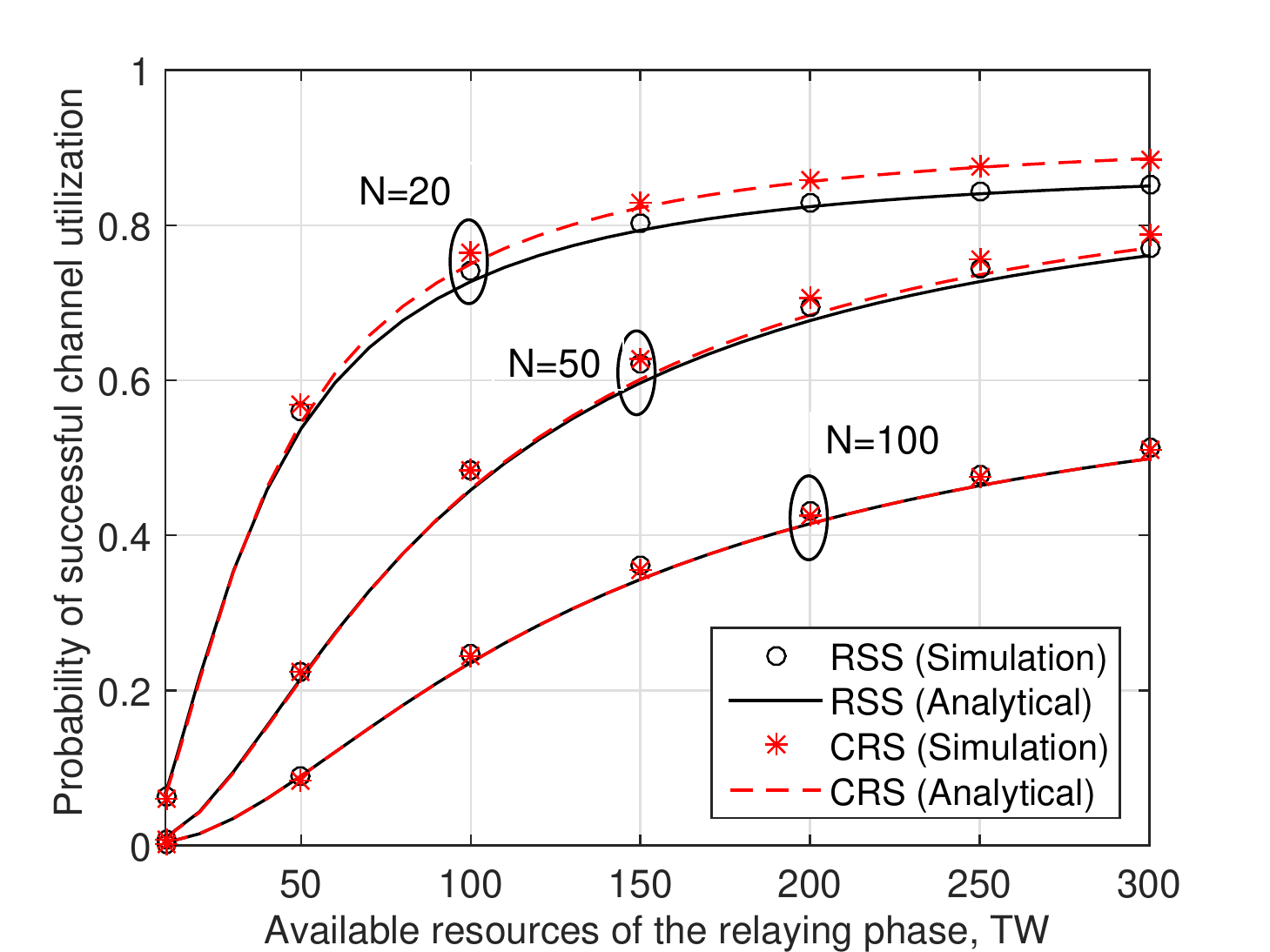}}
\caption{Available resources of the relaying phase, $TW$, versus the considered metrics with different number of channels $N$.}\label{fig_valid}
\end{figure}
\fi
Fig.~\ref{fig_valid} plots the available resources of the relaying phase versus the three considered metrics for the two scheduling schemes with different number of channels of the aggregation phase. Note that we define the available resources as the multiplication of transmission time and bandwidth for the relaying phase in this paper. From all the subfigures, we can see that our derived analytical results provide a good approximation to the exact results. However, unlike Monte Carlo simulations which requires a huge computation time (i.e., one simulation point in Fig.~\ref{metric1} took approximately 14.5 days to generate using Matlab running on a Windows 7 machine with Intel Core i7-4790 processor at 3.6 GHz and 16 GB RAM), our derived analytical expressions allow the metrics to be computed quickly, especially for the RRS scheme.

According to Fig.~\ref{fig_valid}, as the available resources of the relaying phase $TW$ increase, all the metrics increase, which implies that the provision of more resources to the relaying phase can improve the performance. This is due to the fact that adding more resources of the relaying phase can increase the relaying phase's channel success probability for both schemes. From Fig.~\ref{metric3}, we observe that increasing number of channels of the aggregation phase degrades the probability of successful channel utilization. But, in terms of the MTD success probability and average number of successful MTDs, providing more channels for the aggregation phase does not always improve the performance as there are several intersection points in Figs.~\ref{metric1} and~\ref{metric2}. A more detailed discussion about the impact of $N$ will be provided in the following subsection. Additionally, Figs.~\ref{metric1} and~\ref{metric2} almost have the same curve shape. We have tested other system parameter sets and find that $\bar{K}_{\textrm{suc}}\approx \bar{m}\times\psuc$.

By comparing the RRS and CRS schemes, we observe that the RSS scheme performs worse than the CRS scheme when $N$ is small and $TW$ is large. This is because, when more resources are provided to the relaying phase, the overall performance is mainly determined by the channel success probability of the aggregation phase. It has been shown in Fig.~\ref{fig_channeloutage1} that $\apouta^{c}$ for the CRS outperforms than $\apouta^{r}$ for the RRS when the number of channels $N$ is less. Since the complexity of the RRS is much simpler than the CRS and their performance has little difference in most cases, in the following sections, we focus on the RRS scheme and study the effect of the number of channels and the aggregator's density. We also drop the superscript $r$. Note that the observed trends also hold for the CRS.

\subsection{Effect of Number of Channels}\label{sec:result:N}
Fig.~\ref{fig:effectN} plots the number of channels versus probability of successful channel utilization $\achannel$ and MTD success probability $\psuc$ with different $TW$, respectively. As indicated in the above subsection where $\bar{K}_{\textrm{suc}}\approx \bar{m}\times\psuc$, the curves for the MTD success probability are almost the same as average number of successful MTDs; hence, we do not show its figure for the sake of brevity. Fig.~\ref{fig:effectN_channel} shows that as the number of channels of the aggregation phase increases, the probability of successful channel utilization always decreases, which has also been observed in Fig.~\ref{metric3}. This is because as more channels become available for the aggregation phase, there is lesser chance for a channel to be occupied by a MTD. Consequently, the channel utilization performance is deteriorated. In addition, we find that the worst channel performance is when the resources for the relaying phase are very small, since $\achannel$ drops a lot for $TW=50$ compared to the curve for $TW=300$.
\ifCLASSOPTIONpeerreview
\begin{figure}
\centering
\subfigure[Probability of successful channel utilization.]{ \label{fig:effectN_channel}\includegraphics[width=0.44\textwidth]{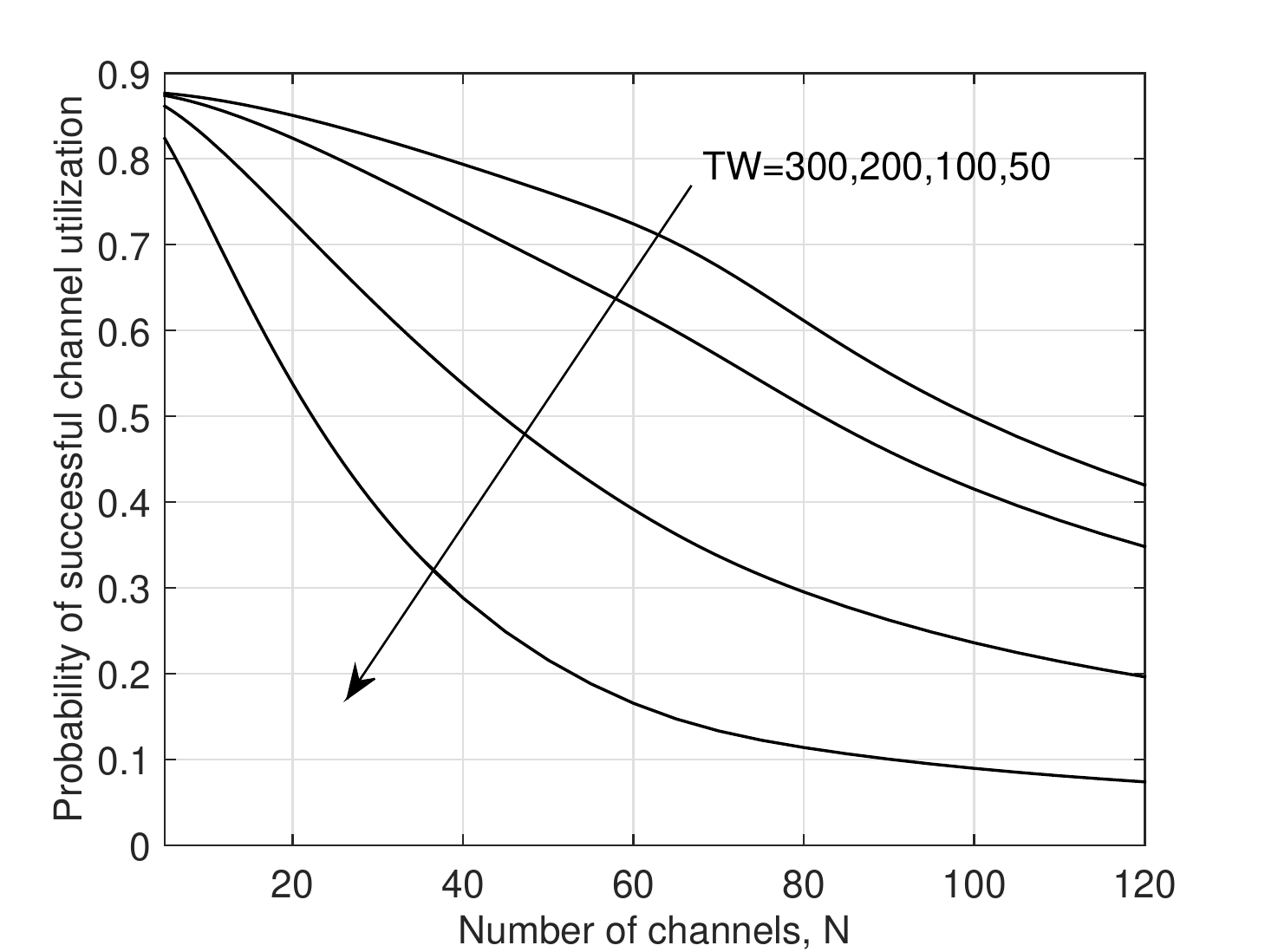}}
\mbox{\hspace{0.5cm}}
\subfigure[MTD success probability.]{\label{fig:effectN_node}\includegraphics[width=0.44\textwidth]{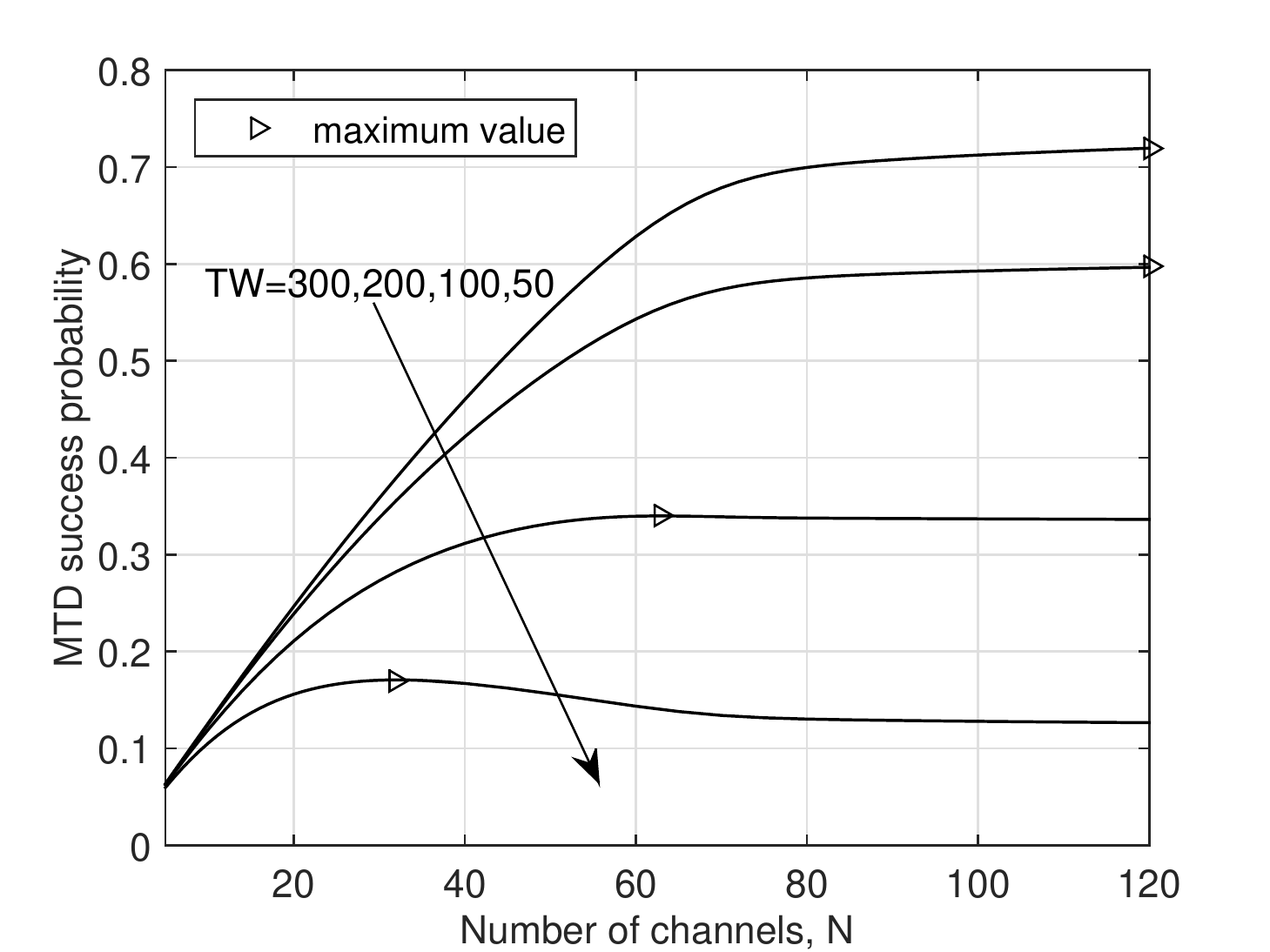}}
\caption{Number of channels, $N$, versus (a) probability of successful channel utilization and (b) MTD success probability.}\label{fig:effectN}
\end{figure}
\else
\begin{figure}
\centering
\subfigure[Probability of successful channel utilization.]{ \label{fig:effectN_channel}\includegraphics[width=0.41\textwidth]{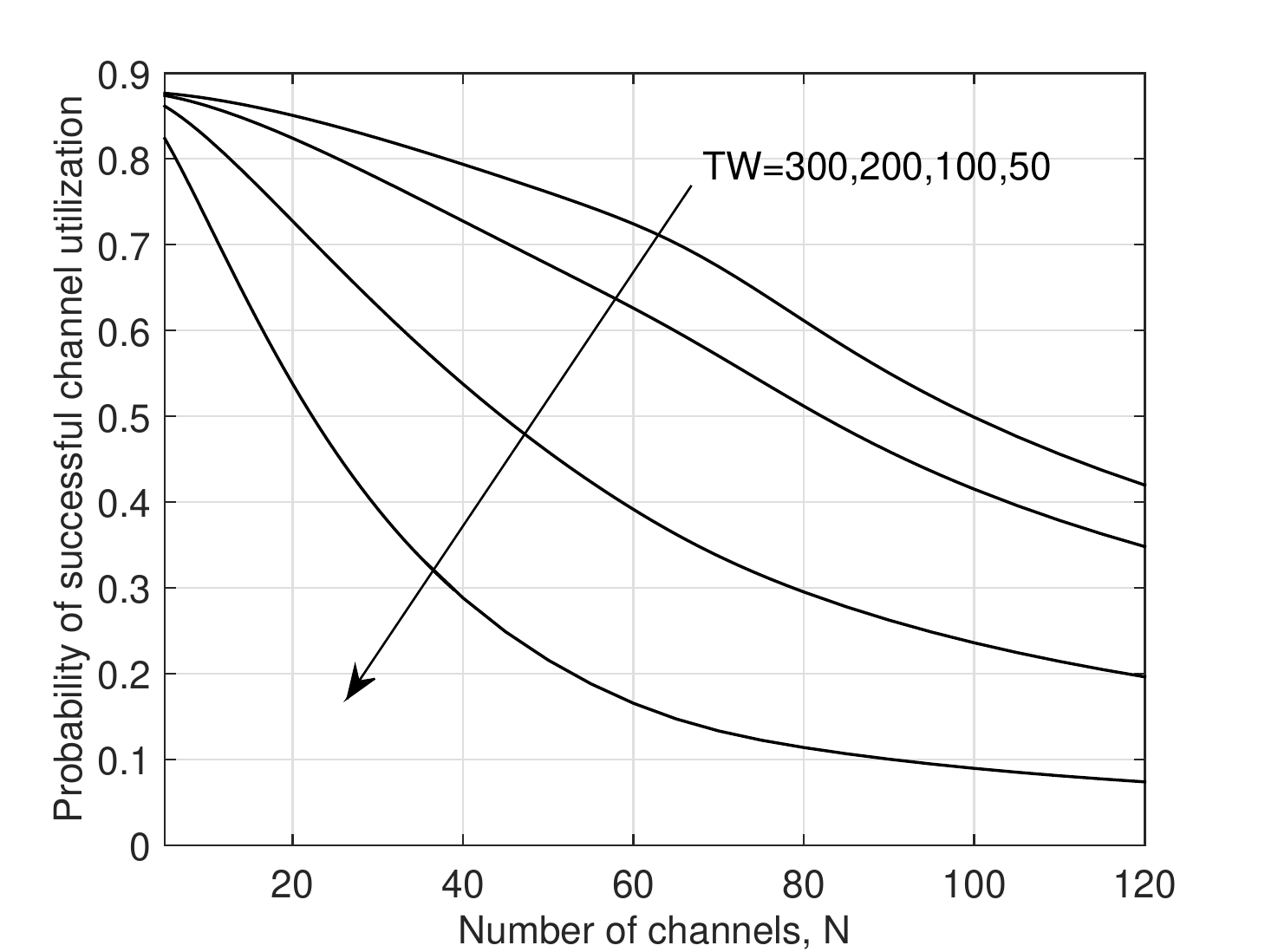}}
\subfigure[MTD success probability.]{\label{fig:effectN_node}\includegraphics[width=0.41\textwidth]{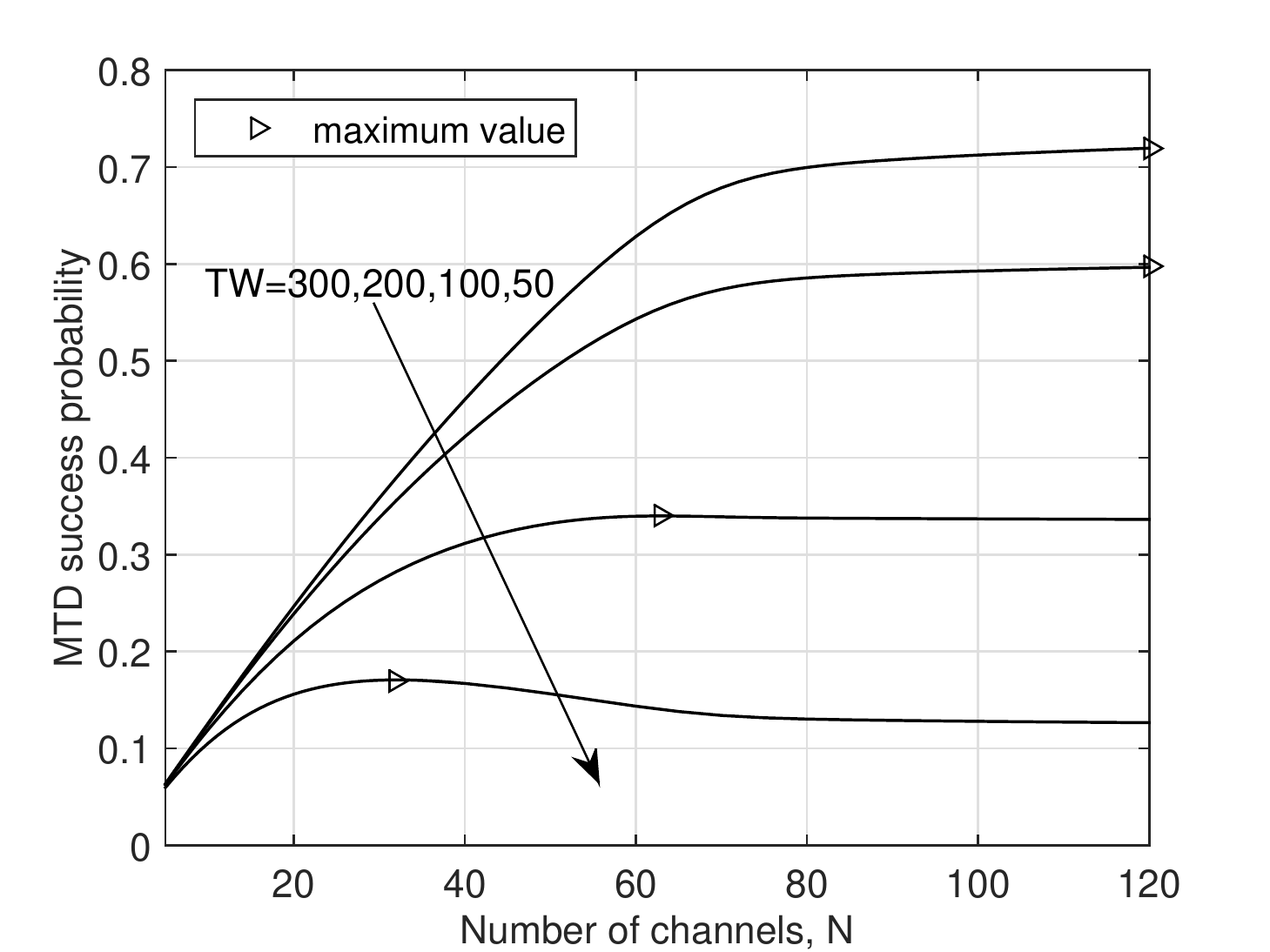}}
\caption{Number of channels, $N$, versus (a) probability of successful channel utilization and (b) MTD success probability.}\label{fig:effectN}
\end{figure}
\fi

In Fig.~\ref{fig:effectN_node}, we mark the maximum MTD success probability for each curve within the considered range. It shows that, when $TW$ becomes large, the MTD success probability increases with the increasing number of channels. After a certain value of $N$, the curves are almost flat. For example, under $TW=300$, by increasing $N$ from $80$ to $100$, $\psuc$ only increases from 0.6997 to 0.7195 while $\achannel$ decreases a lot (i.e., from 0.6113 to 0.4197). This manifests that the performance gain achieved by adding more channels is very little and also the channels are greatly underused. When the available resources for the relaying phase are small (i.e., for the curves with $TW=100$ or $50$), the MTD success probability increases at first and then decreases as $N$ increases. This can be explained as follows. From the average perspective, according to~\eqref{eq:summary_channel1}, the MTD success probability is determined by three factors. As $N$ increases, the average non-drop probability of a MTD $\apndrop$ increases at first and then stays as 1. The channel success probability of the aggregation phase, $\apouta$, almost keeps as a constant at first since each channel is fully occupied. $\apouta$ then begins to increase because the channel occupation probability decreases, which reduces the interference experienced at each channel. As for the channel success probability of the relaying phase $\apoutb$, it is always decreasing. This is because, as $N$ increases, more data is required to be transmitted by data aggregator thereby degrading the channel performance. For the first two factors, they are independent of the available resources of the relaying phase. However, $\apoutb$ strongly relies on the relaying phase's resources. We find that when the available resources of the relaying phase are scarce (i.e., $TW$ is small), $\apoutb$ drops a lot as $N$ increases, compared the case where $TW$ is large. The interplay of these factors results in the trends as shown in Fig.~\ref{fig:effectN_node}.

\ifCLASSOPTIONpeerreview
\begin{figure}
\centering
\subfigure[Different path-loss exponent.]{\label{fig:effectN_alpha}\includegraphics[width=0.44\textwidth]{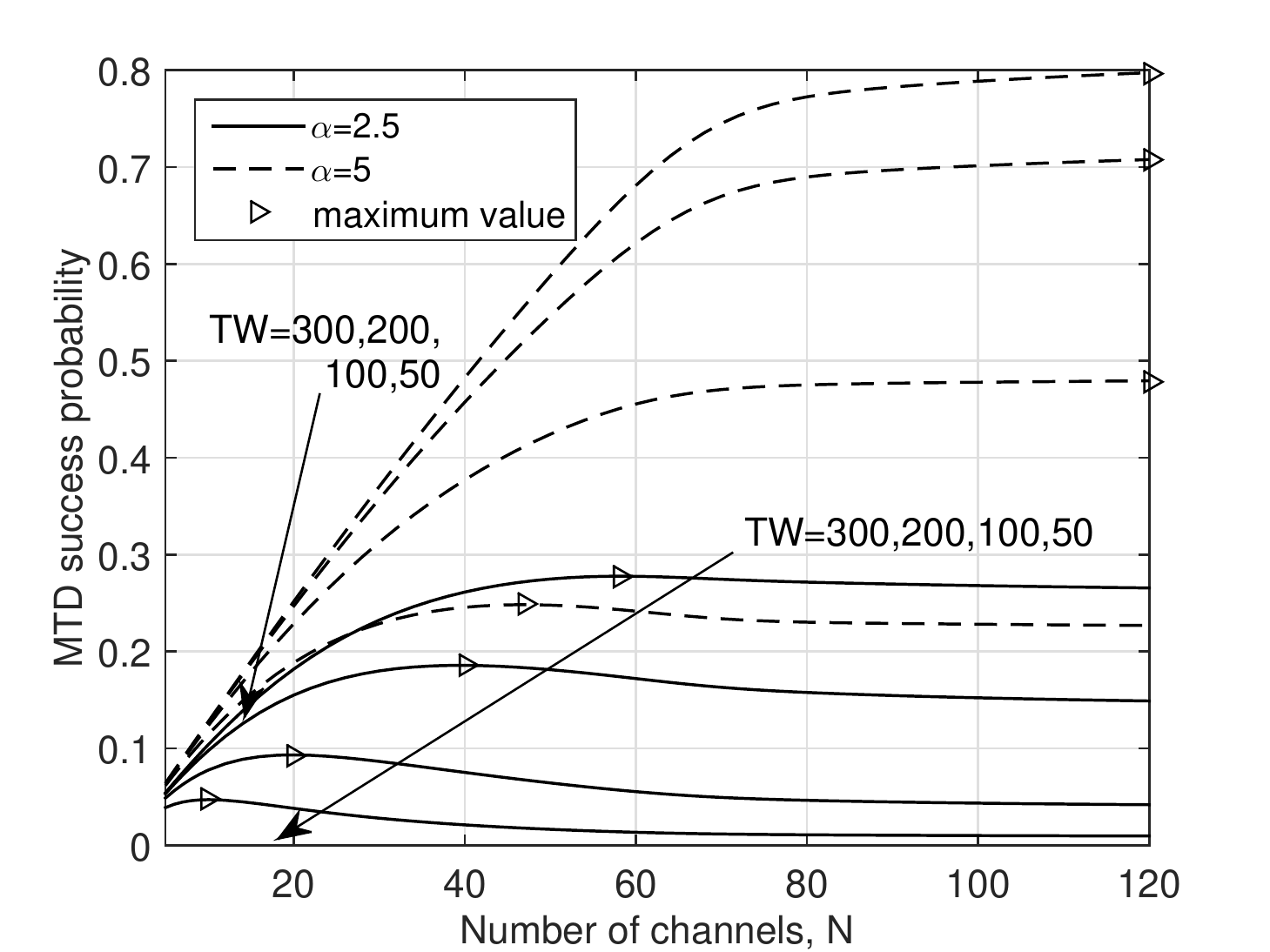}}
\mbox{\hspace{0.5cm}}
\subfigure[Different aggregator's density.]{ \label{fig:effectN_density}\includegraphics[width=0.44\textwidth]{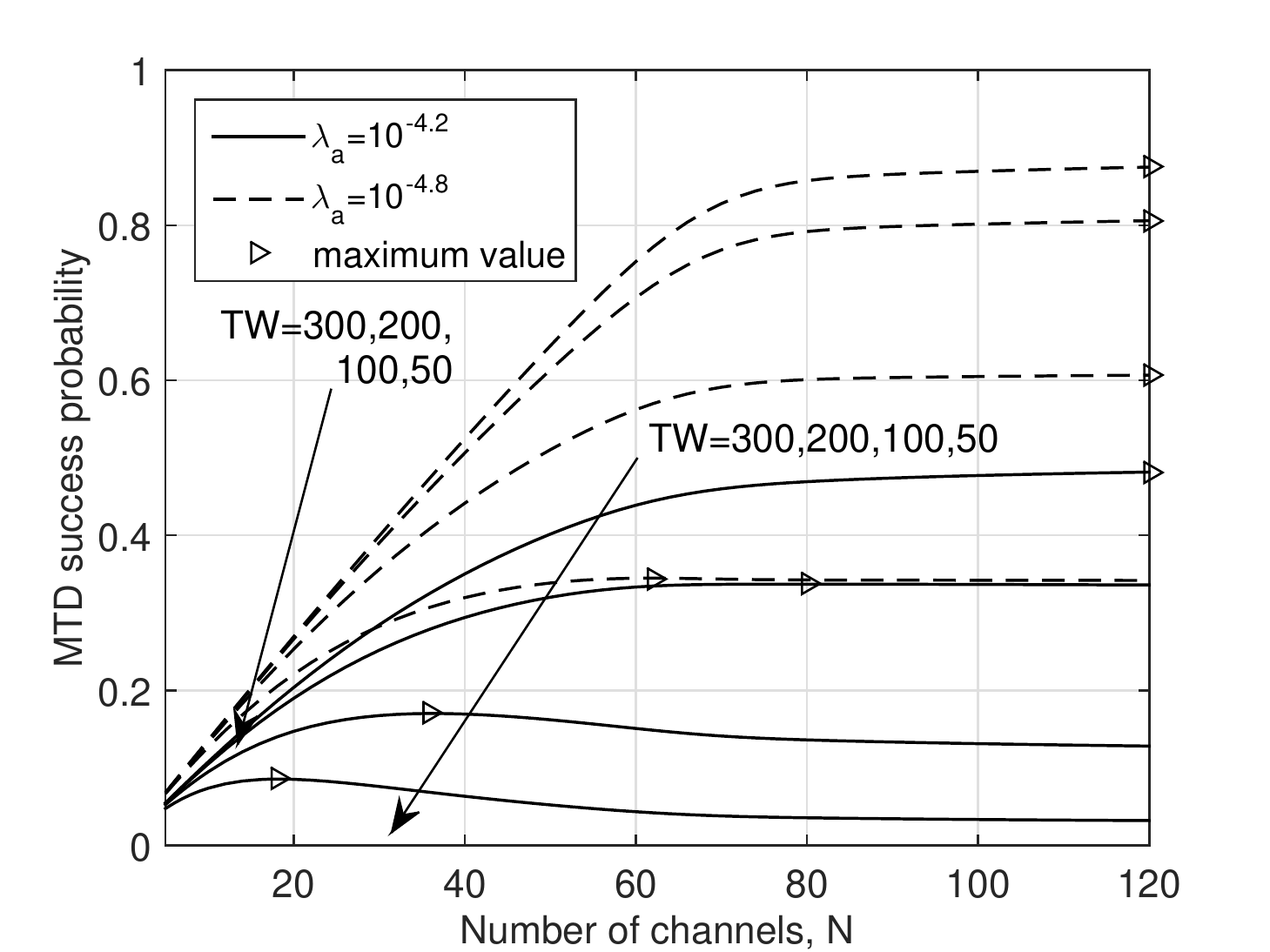}}
\caption{Number of channels, $N$, versus the MTD success probability with (a) different path-loss exponent and (b) different aggregator's density.}\label{fig:effectN_diff}
\end{figure}
\else
\begin{figure}
\centering
\subfigure[Different path-loss exponent.]{\label{fig:effectN_alpha}\includegraphics[width=0.41\textwidth]{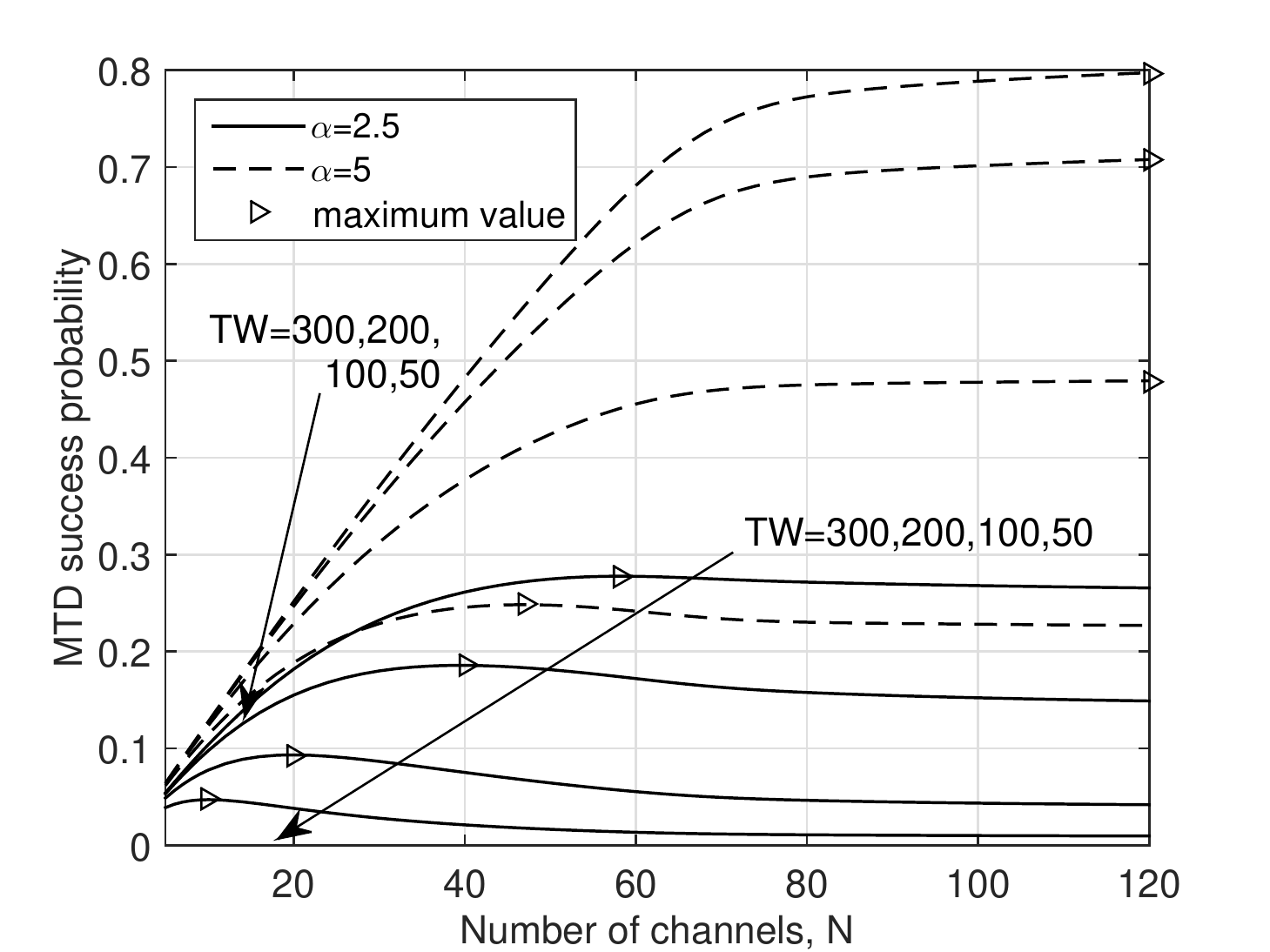}}
\subfigure[Different aggregator's density.]{ \label{fig:effectN_density}\includegraphics[width=0.41\textwidth]{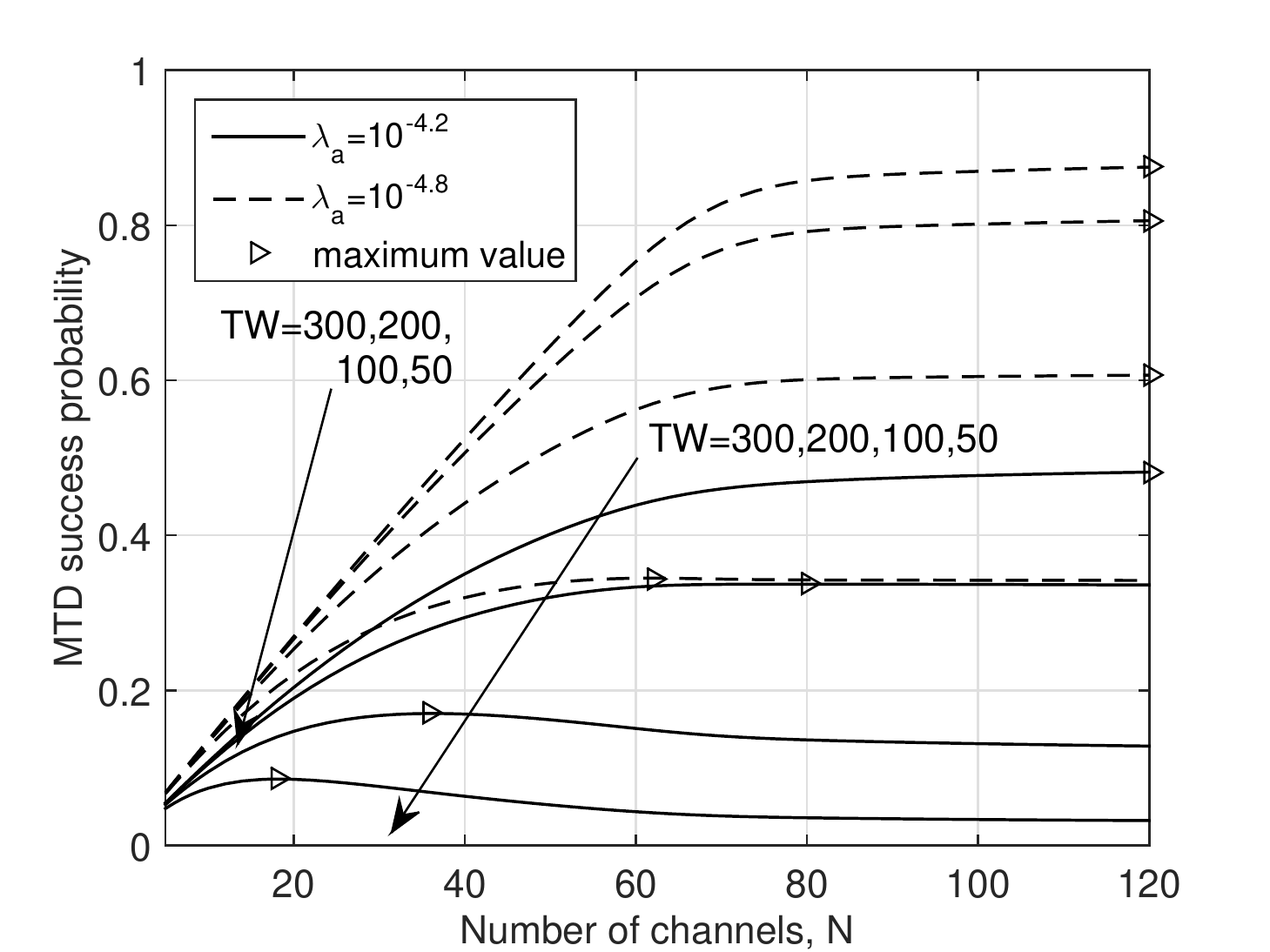}}
\caption{Number of channels, $N$, versus the MTD success probability with (a) different path-loss exponent and (b) different aggregator's density.}\label{fig:effectN_diff}
\end{figure}
\fi
We have also examined the impact of $N$ under different path-loss exponent and different aggregator's density as shown in Fig.~\ref{fig:effectN_diff}. From these figures, we observe that when the path-loss exponent is small or the aggregator's density is large, such non-monotonic trend is more likely to occur. This is because, when $N$ gets larger, under these scenarios, the interference experienced at the relaying phase is more severe, which causes $\apoutb$ to decrease faster. Thus, we see that increasing the number of channels of the aggregation phase does not always benefit to the mMTC, especially for the case when the channel performance of the relaying phase is poor.

\subsection{Effect of Density of MTDs and Aggregators}\label{sec:effect:dense}
In this subsection, we analyze the effect of MTDs and aggregator's density. As indicated in Section~\ref{subsec:system:network}, the density of MTDs is related to the density of aggregators by $\lambda_{m}=\bar{m}\lambda_{a}$. We fix $\bar{m}$ here. Fig.~\ref{fig:effectdenisty_channel} and Fig.~\ref{fig:effectdenisty_node} plot the aggregator's density versus the probability of successful channel utilization and MTD success probability with different number of channels, respectively. As illustrated in these subfigures, as the aggregator's density (equivalently, the MTD's density) increases, both the probability of successful channel utilization and MTD success probability decrease, which implies that the performance for a single MTD, aggregator and channel worsens.

This can be explained as follows. As indicated in Remark 1, the channel occupation probability and non-drop probability of a MTD keep the same regardless of aggregator's density. Thus, these two metrics only rely on the channel success probability of the aggregation and relaying phases (i.e., $\apouta$ and $\apoutb$). When $\lambda_a$ increases, more MTDs are brought into the network and more interference will be generated on each channel, which results in decreasing $\apouta$. With respect to $\apoutb$, it decreases at first and then increases. Initially, because more aggregators will be associated to the same BS, less resources are allocated to each aggregator and $\apoutb$ decreases. After a certain point (i.e., when the performance for $\apouta$ is very poor), less MTDs can be successfully connected to the aggregator. Hence, $\apoutb$ begins to increase since the data transmitted by an aggregator is significantly reduced. Note that, when aggregators are very dense, the performance of $\apouta$ is so severe such that the overall performance is mainly governed by $\apouta$. Consequently, the general trends for $\psuc$ (equivalently, $\bar{K}_{\textrm{suc}}$) and $\achannel$ are decreasing.
\ifCLASSOPTIONpeerreview
\begin{figure}
\centering
\subfigure[Probability of successful channel utilization.]{ \label{fig:effectdenisty_channel}\includegraphics[width=0.44\textwidth]{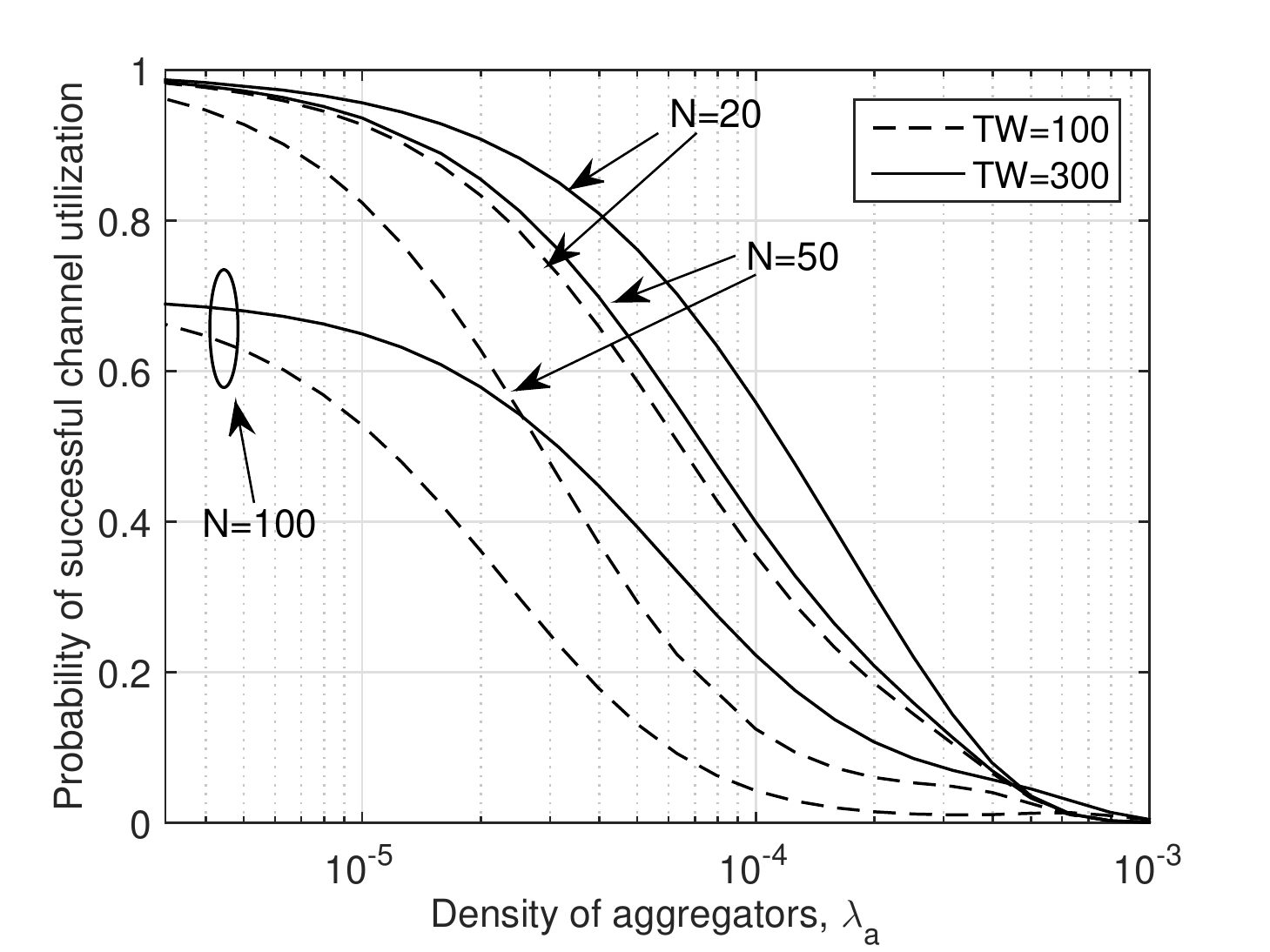}}
\mbox{\hspace{0.5cm}}
\subfigure[MTD success probability.]{\label{fig:effectdenisty_node}\includegraphics[width=0.44\textwidth]{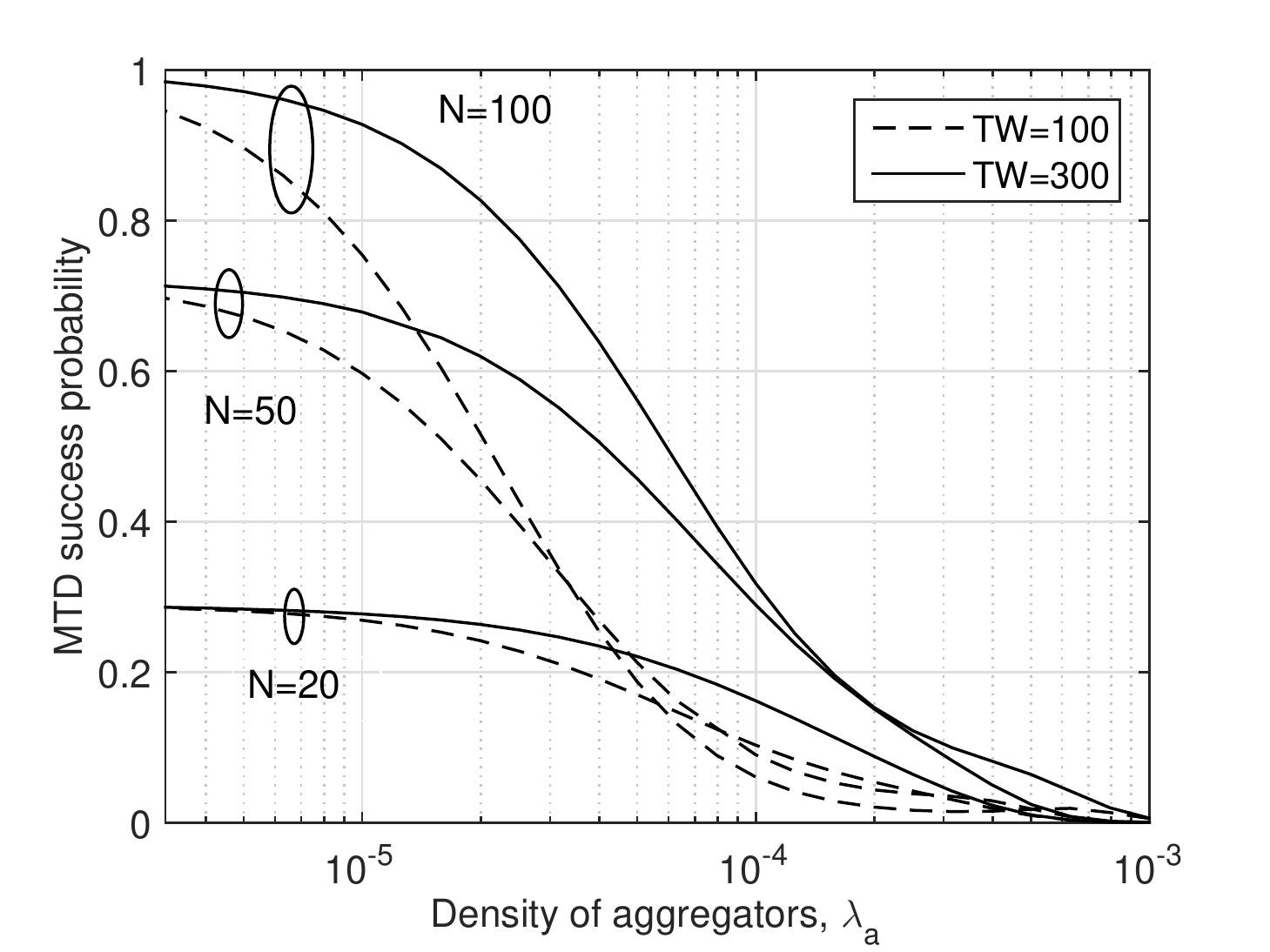}}

\subfigure[Fixed average number of MTDs per aggregator.]{\label{fig:effectdenisty_nodeAll1}\includegraphics[width=0.44\textwidth]{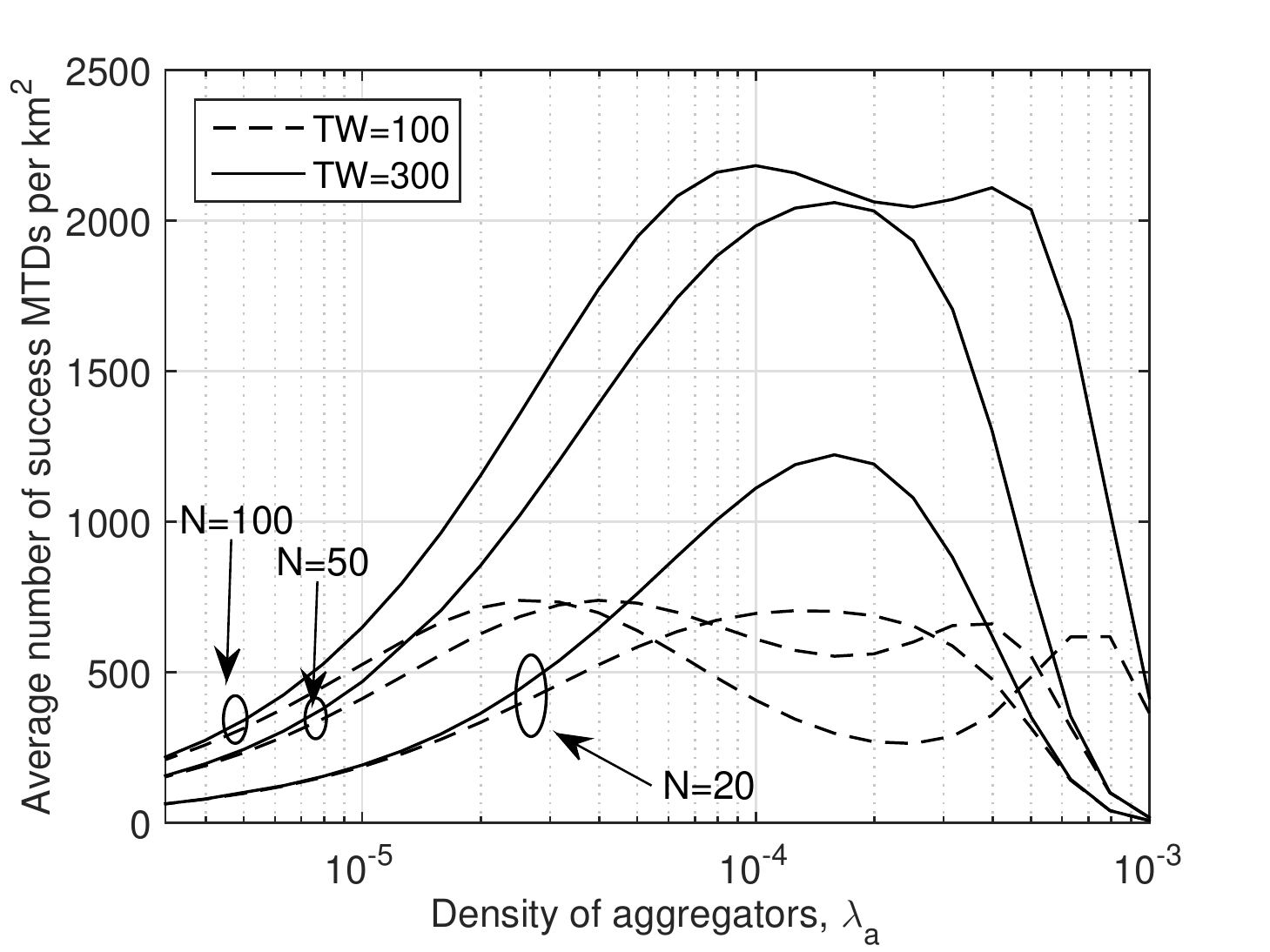}}

\caption{Aggregator's density, $\lambda_a$, versus (a) probability of successful channel utilization, (b) MTD success probability and (c) average number of successful MTDs per km$^2$.}\label{fig:effectdensity}
\end{figure}
\else
\begin{figure}
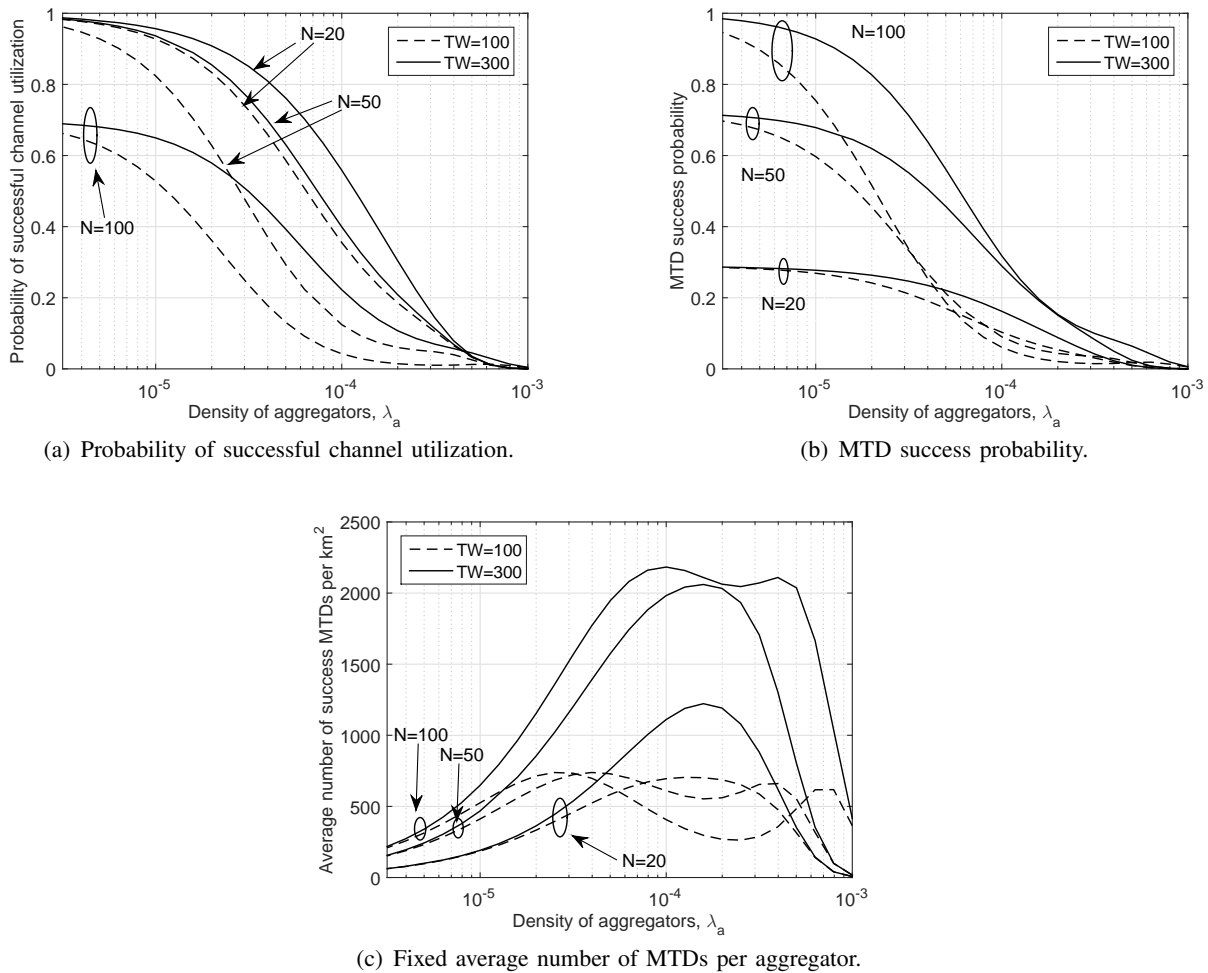

\centering
\subfigure[Probability of successful channel utilization.]{ \label{fig:effectdenisty_channel}\includegraphics[width=0.41\textwidth]{effect_density_channel}}
\subfigure[MTD success probability.]{\label{fig:effectdenisty_node}\includegraphics[width=0.41\textwidth]{fig/effect_density_MTD}}

\subfigure[Fixed average number of MTDs per aggregator.]{\label{fig:effectdenisty_nodeAll1}\includegraphics[width=0.41\textwidth]{effect_density_nodeAll}}

\caption{Aggregator's density, $\lambda_a$, versus (a) probability of successful channel utilization, (b) MTD success probability and (c) average number of successful MTDs per km$^2$.}\label{fig:effectdensity}
\end{figure}
\vspace{-2mm}
\fi

We also examine the average number of successful MTDs per km$^2$. Fig.~\ref{fig:effectdenisty_nodeAll1} plots the aggregator's density versus the average number of successful MTDs per km$^2$. It is shown that, when the number of channels of the aggregation phase is not that many and the relaying phase's available resources $TW$ are relatively large, the average number of successful MTDs per km$^2$ increases at the first and then decreases with the increasing aggregator's density. General speaking, increasing the aggregator's density first improves the average number of successful MTDs per km$^2$, since at the low density region the performance is mainly impacted by $\lambda_a$. When the aggregators are too dense such that increasing $\lambda_a$ cannot compensate the loss at each aggregator, the overall performance deteriorates. In summary, properly choosing the aggregator's density can enhance the performance in terms of the average number of successful MTDs per km$^2$.

\section{Conclusions}\label{sec:conclusion}
In this work, we have proposed a tractable analytical framework to study a two-phase cellular-base mMTC system, where data from MTDs is first aggregated at the data aggregator and then relayed to BSs. The aggregators also schedule the limited resources among their associated MTDs. Two scheduling schemes are considered, i.e., the RRS and CRS schemes. We derived the SIR distribution for each phase and we also investigated three metrics to evaluate the overall mMTC performance: the MTD success probability, average number of successful MTDs and probability of successful channel utilization. The accuracy of the derived results is confirmed by simulations. Our results showed that, compared to the CRS scheme, the RRS scheme can achieve almost the same overall mMTC performance as long as the aggregation aggregation's recourse is not very limited. In addition, providing more numbers of channels at the aggregation phase can sometimes degrade overall mMTC performance, especially when the relaying phase experiences severe interference. \redcom{Future work can consider the non-orthogonal multiple access for MTDs and aggregators, as envisaged in~\cite{Mahyar-2017}.}
\section*{Appendix A: Proof of Lemmas 1 and 2}\label{appendix:prob}
\begin{proof}
The channel occupation probability is independent of the resource scheduling scheme and it only relies on the number of MTDs requiring data transmission. Given that there are $K$ MTDs requiring data transmission for a typical aggregator, the available $N$ channels will be assigned to $\textrm{min}\left\{K,N\right\}$ MTDs and the remaining $K-\textrm{min}\left\{K,N\right\}$ MTDs are dropped. Hence, the conditional channel occupation probability is given by $p_O=\frac{\textrm{min}\left\{K,N\right\}}{N}$. After averaging the distribution of $K$ (i.e., a Poisson distribution with mean $\bar{m}$), we have the average channel occupation probability as
\ifCLASSOPTIONpeerreview
\begin{align}
 \bar{p}_O=&\mathbb{E}_{K}\left\{\frac{\textrm{min}\left\{K,N\right\}}{N} \right\}=\sum_{k=0}^{N}\frac{k}{N}\frac{\bar{m}^k\exp(-\bar{m})}{k!}+\sum_{k=N+1}^{\infty}\frac{\bar{m}^k\exp(-\bar{m})}{k!}\nonumber\\
 =&1-\frac{\Gamma\left[1+N,\bar{m}\right]}{\Gamma\left[1+N\right]}+\frac{\bar{m}\Gamma[1+N,\bar{m}]-\exp(-\bar{m})\bar{m}^{N+1}}{N^2\Gamma[N]}.
\end{align}
\else
\begin{align}
 \bar{p}_O=&\mathbb{E}_{K}\left\{\frac{\textrm{min}\left\{K,N\right\}}{N} \right\}\nonumber\\
 =&\sum_{k=0}^{N}\frac{k}{N}\frac{\bar{m}^k\exp(-\bar{m})}{k!}+\sum_{k=N+1}^{\infty}\frac{\bar{m}^k\exp(-\bar{m})}{k!}\nonumber\\
 =&1-\frac{\Gamma\left[1+N,\bar{m}\right]}{\Gamma\left[1+N\right]}+\frac{\bar{m}\Gamma[1+N,\bar{m}]-\exp(-\bar{m})\bar{m}^{N+1}}{N^2\Gamma[N]}.
\end{align}
\fi

Likewise, under the RRS, by conditioning on the number of MTDs attached to a typical aggregator, we can have the conditional non-drop probability for a MTD given by $\pndrop=\frac{N}{\textrm{max}\left\{K,N\right\}}$. With regards to the CRS, when $K\leq N$, a MTD will always be scheduled a channel (i.e., $\pndrop=1$). For the case when $K>N$, since the fading gains among these $K$ MTDs are identically and independently distributed, the probability that a MTD's fading gain can be ranked as the top $K$ best fading gains among $N$ MTDs is $\frac{K}{N}$ (equivalently, $\pndrop=\frac{K}{N}$). Note that the conditional non-drop probability is exactly the same as the one for the RRS.

We then average this conditional probability over the distribution of $K$ and obtain the average non-drop probability as
\ifCLASSOPTIONpeerreview
\begin{align}
\apndrop=&\mathbb{E}_{K}\left\{\frac{N}{\textrm{max}\left\{K,N\right\}}\right\}=\sum_{k=0}^{N}\frac{\bar{m}^k\exp(-\bar{m})}{k!}+\sum_{k=N+1}^{\infty}\frac{N}{k}\frac{\bar{m}^k\exp(-\bar{m})}{k!}\nonumber\\
 =&\frac{\Gamma[1+N,\bar{m}]}{\Gamma[1+N]}+\exp(-\bar{m})\bar{m}^{1+N}N\frac{\,_2F_2\left[\left\{1,1+N\right\},\left\{2+N,2+N\right\},\bar{m}\right]}{(N+1)(N+1)!}.
\end{align}
\else
\begin{align}
&\apndrop=\mathbb{E}_{K}\left\{\frac{N}{\textrm{max}\left\{K,N\right\}}\right\}\nonumber\\
&=\sum_{k=0}^{N}\frac{\bar{m}^k\exp(-\bar{m})}{k!}+\sum_{k=N+1}^{\infty}\frac{N}{k}\frac{\bar{m}^k\exp(-\bar{m})}{k!}\nonumber\\
 &=\frac{\Gamma[1+N,\bar{m}]}{\Gamma[1+N]}\nonumber\\
 &+\exp(-\bar{m})\bar{m}^{1+N}N\frac{\,_2F_2\left[\left\{1,1+N\right\},\left\{2+N,2+N\right\},\bar{m}\right]}{(N+1)(N+1)!}.
\end{align}
\fi

\end{proof}

\section*{Appendix B: Proof of Theorem~\ref{theo:channeloutage1}}\label{appendix:outage1}
\begin{proof}
For a certain channel of a typical aggregator, given it is occupied by a MTD, the instantaneous SIR is given by
\begin{align}
\textsf{SIR}_1=\frac{h}{\sum_{x\subset\Phi_{\textrm{MTD}}^{\redcom{\textrm{served}}}}g r_s^{\alpha}x^{-\alpha}},
\end{align}
\noindent where $x$ denotes both the location and the interfering MTD which occupies the certain channel, $r_s$ is the distance between the MTD and its serving \redcom{aggregator}, $h$ and $g$ are the fading power gain on the desired link and interfering link which follow a gamma distribution with shape parameter $m_1$ and an exponential distribution, respectively.

Since $\Phi_{\textrm{MTD}}^{\redcom{\textrm{served}}}$ is the HPPP with density $\bar{p}_O\adensity$, using the Campbell theorem~\cite{Haenggi-2012}, we obtain the average channel success probability at the typical aggregator
\begin{align}
\apouta^{r}=&\Pr\left(\textsf{SIR}_1\geq\gamma_1\right)=\Pr\left(\frac{h}{\sum_{x\subset\Phi_{\textrm{MTD}}^{\redcom{\textrm{served}}}}g r_s^{\alpha}x^{-\alpha}}\geq\gamma_1\right)\nonumber\\
=&\sum\limits_{t=0}\limits^{m_1-1}\frac{(-s)^{t}}{t!}\left.\frac{d^t}{d s^t}\mathcal{M}_{I_{1}}(s)\right|_{s=m_1\gamma_1},
\end{align}
\noindent where
\begin{align}\label{eq:MGFintf1}
 \mathcal{M}_{I_{1}}(s)=&\mathbb{E}_{\Phi_{\textrm{MTD}}^{\redcom{\textrm{served}}},g,r_s}\left\{\exp\left(-s\sum_{x\subset\Phi_{\textrm{MTD}}^{\redcom{\textrm{served}}}}g r_s^{\alpha}x^{-\alpha}\right)\right\}\nonumber\\
 =&\exp\left(-\bar{p}_O\adensity\pi\mathbb{E}_{r_s}\!\left\{r_s^2\right\}\mathbb{E}_g\left\{g^{\frac{2}{\alpha}}\right\}\Gamma\!\left[1-\frac{2}{\alpha}\right]s^{\frac{2}{\alpha}}\right)\nonumber\\
 =&\exp\left(-\bar{p}_O\adensity\pi\frac{R_s^2}{2}\Gamma\!\left[1+\frac{2}{\alpha}\right]\Gamma\!\left[1-\frac{2}{\alpha}\right]s^{\frac{2}{\alpha}}\right).
\end{align}
\end{proof}

\section*{Appendix C: Proof of Proposition~\ref{theo:channeloutage1p}}\label{appendix:outage1p}
\begin{proof}
The channel success probability experienced at an aggregator, under the CRS, relies strongly on the number of MTDs within this aggregator's serving zone. Let $\pouta^{c}(K)$ denote the conditional channel success probability given $K$ MTDs requiring data transmission for a typical aggregator.

When $K\leq N$, given a certain channel is occupied by a MTD, the situation is the same as for the RRS. Hence, $\pouta^{c}(K)=\apouta^{r}$ when $K\leq N$.

When $K>N$, only $N$ MTDs with better fading gains will be allowed to occupy the channels among $K$ MTDs. For a typical aggregator which is assumed to be located at the origin, given a certain channel is scheduled to a MTD with the $i$-th best fading gain (i.e., $h_{(i),K}$), the channels success probability is $\pouta^{c}(K,i)=\Pr\left(\frac{h_{(i),K}}{I_1}\geq\gamma_1\right)$. The distribution of $h_{(i),K}$ is
\ifCLASSOPTIONpeerreview
\begin{align}
f(h_{(i),K})\!=\!\frac{K!}{(i\!-\!1)!(K\!-\! i)!}\!\left(\frac{\Gamma\![m_1,0,m_1h_{(i),k}]}{\Gamma\![m_1]}\right)^{\!\!i-1}\!\!\left(\frac{\Gamma\![m_1,m_1h_{(i),k}]}{\Gamma\![m_1]}\right)^{\!\!K-i}
\frac{m_1^{m_1}h_{(i),k}^{m_1-1}\exp(-m_1h_{(i),k})}{\Gamma\![m_1]}.
\end{align}
\else
\begin{align}
&f(h_{(i),K})\!=\!\frac{K!}{(i\!-\!1)!(K\!-\! i)!}\!\left(\frac{\Gamma\![m_1,0,m_1h_{(i),k}]}{\Gamma\![m_1]}\right)^{\!\!i-1}\nonumber\\
&\times\left(\frac{\Gamma\![m_1,m_1h_{(i),k}]}{\Gamma\![m_1]}\right)^{\!\!K-i}
\frac{m_1^{m_1}h_{(i),k}^{m_1-1}\exp(-m_1h_{(i),k})}{\Gamma\![m_1]}.
\end{align}
\fi
Note that this distribution is obtained using the order statistical theory~\cite{David-1970} and the fact that the original distribution of fading gain $h$ is a gamma distribution. Unlike the RRS where the fading gain has a nice distribution allowing us to easily compute the channel success probability, the distribution of $h_{(i),K}$ is very complicated. Instead, we use the Gil-Pelaez inversion theorem to work out the channel success probability.

According to the Gil-Pelaez inversion theorem~\cite{Pelaez-1951,Renzo-2014}, we can rewrite the conditional channel success probability as
\ifCLASSOPTIONpeerreview
\begin{align}
&\pouta^{c}(K,i)=\Pr\left(I_1\leq\frac{h_{(i),K}}{\gamma_1}\right)\nonumber\\
&=\frac{1}{2}-\frac{1}{\pi}\int_{0}^{\infty}\frac{1}{w}\textrm{Im}\left\{\mathcal{M}_{I_1}(-\ii w)\mathbb{E}_{h_{(i),K}}\left\{\exp\left(-\ii w\frac{h_{(i),K}}{\gamma_1}\right)\right\}\right\}\textup{d}w \nonumber\\
&\approx\frac{1}{2}\!-\!\frac{1}{\pi}\int_0^{\infty}\!\!\frac{1}{w}\textrm{Im}\!\left\{\exp\!\left(\!-\bar{p}_O\adensity\pi\frac{R_s^2}{2}\Gamma\!\left[\frac{\alpha\!+\!2}{\alpha}\right]\Gamma\!\left[\frac{\alpha\!-\!2}{\alpha}\right]
 \!\left(-\ii w\right)^{\!\frac{2}{\alpha}}\right)\exp\!\left(-\ii w\frac{\mathbb{E}_{h_{(i),k}}\left\{h_{(i),k}\right\}}{\gamma_1}\!\right)\!\right\}\textup{d}w,
\end{align}
\else
\begin{align}
&\pouta^{c}(K,i)=\Pr\left(I_1\leq\frac{h_{(i),K}}{\gamma_1}\right)\nonumber\\
&=\frac{1}{2}-\frac{1}{\pi}\nonumber\\
&\times\int_{0}^{\infty}\!\frac{1}{w}\textrm{Im}\left\{\mathcal{M}_{I_1}(-\ii w)\mathbb{E}_{h_{(i),K}}\left\{\exp\left(-\ii w\frac{h_{(i),K}}{\gamma_1}\right)\right\}\!\right\}\textup{d}w \nonumber\\
&\approx\frac{1}{2}\!-\!\frac{1}{\pi}\int_0^{\infty}\!\!\frac{1}{w}\textrm{Im}\!\left\{\exp\!\left(-\ii w\frac{\mathbb{E}_{h_{(i),k}}\left\{h_{(i),k}\right\}}{\gamma_1}\!\right)\!\right.\nonumber\\
&\left.\times\exp\!\left(\!-\bar{p}_O\adensity\pi\frac{R_s^2}{2}\Gamma\!\left[\frac{\alpha\!+\!2}{\alpha}\right]\Gamma\!\left[\frac{\alpha\!-\!2}{\alpha}\right]
 \!\left(-\ii w\right)^{\!\frac{2}{\alpha}}\right)\right\}\textup{d}w,
\end{align}
\fi
\noindent where the second step comes from the Gil-Pelaez inversion theorem, and the last step leverages the $\mathcal{M}_{I_1}(s)$ derived in~\eqref{eq:MGFintf1}. Note that the approximation in the last step comes from the application of Jensen inequality for the term $\mathbb{E}_{h_{(i),K}}\left\{\exp\left(-\ii w\frac{h_{(i),K}}{\gamma_1}\right)\right\}$ (i.e., $\mathbb{E}_{h_{(i),K}}\left\{\exp\left(-\ii w\frac{h_{(i),K}}{\gamma_1}\right)\right\}\approx \exp\!\left(-\ii w\frac{\mathbb{E}_{h_{(i),k}}\left\{h_{(i),k}\right\}}{\gamma_1}\!\right)$). Such an approximation reduces the computation complexity while maintaining the accuracy of the final results as shown in Fig.~\ref{fig_channeloutage1}.

For a certain channel of a typical aggregator, the fading on this channel can experience from the best fading until the $N$-th best fading with equal probability. Hence, we need to average the above channel success probability and obtain the average channel success probability on the certain channel, given $K$ MTDs, as
\begin{align}
&\pouta^{c}(K)\!\approx\!\frac{\sum\limits_{i=1}^{N}\!\left(\frac{1}{2}\!-\!\frac{1}{\pi}\int_0^{\infty}\!\!\frac{1}{w}\textrm{Im}\!\left\{\!\exp\!\left(\!-\bar{p}_O\adensity\pi\frac{R_s^2}{2}\Gamma\!\left[\frac{\alpha\!+\!2}{\alpha}\right]\Gamma\!\left[\frac{\alpha\!-\!2}{\alpha}\right]
 \!\left(-\ii w\right)^{\!\frac{2}{\alpha}}\right)\!\exp\!\left(\!-\ii w\frac{\mathbb{E}_{h_{(i),k}}\left\{h_{(i),k}\right\}}{\gamma_1}\!\right)\!\!\right\}\textup{d}w\!\!\right)}{N}.
\end{align}

Finally, after averaging $\pouta^{c}(K)$ over the distribution of the number of MTD requiring data transmission, we arrive at the result in Proposition~\ref{theo:channeloutage1p}.
\end{proof}

\section*{Appendix D: Proof of Corollary~\ref{coro:distribution1}}\label{appendix:distribution1}
\begin{proof}
The distribution of the number of active channels $K_1$ strongly relies on i) the number of MTDs $K$ requiring data transmission for the typical aggregator; and ii) whether or not each channel is experiencing channel outage. It is important to note that the latter event is not independent for each channel and the interference on each channel is correlated. Although within the same serving zone the served MTDs will generate different interference on different channels, they all gather around the aggregator such that the interference from a cluster generated on the channel set $\mathcal{N}$ is spatially-correlated. Solving this spatial-correlation is still a challenging open problem. Instead, in this work, we ignore the correlation and adopt the independent interference assumption to work out the distribution. As will be shown in Section~\ref{sec:result}, the numerical results on the performance metrics derived using this approximated distribution are accurate.

By adopting the independent assumption, given the number of MTDs requiring data transmission $K$, the conditional distribution of $K_1$ is in fact following the binomial distribution. Thus the conditional PMF of $K_1$ is given by
\begin{align}\label{eq:distK1}
\Pr\!\,^{r}(K_1=k_1|K)\approx\begin{cases}
              0, & {K<k_1;} \\
            \binom{K}{k_1}(\apouta^{r})^{k_1}(1-\apouta^{r})^{K-k_1}, & {k_1\leq K\leq N;}\\
            \binom{N}{k_1}(\apouta^{r})^{k_1}(1-\apouta^{r})^{N-k_1}, & {K> N;}
             \end{cases}
\end{align}

After averaging the above conditional distribution, we obtain the distribution of the number of active channels for the typical aggregator as
\ifCLASSOPTIONpeerreview
\begin{align}
&\Pr\!\,^{r}(K_1=k_1)=\mathbb{E}_{K}\left\{\Pr\!\,^{r}(K_1=k_1|K)\right\} \nonumber\\
\approx&\sum_{k=k_1}^{N}\binom{k}{k_1}(\apouta^{r})^{k_1}(1-\apouta^{r})^{k-k_1}\frac{\bar{m}^k\exp(-\bar{m})}{k!}+\sum_{k=N+1}^{\infty}\binom{N}{k_1}(\apouta^{r})^{k_1}(1-\apouta^{r})^{k-k_1}\frac{\bar{m}^k\exp(-\bar{m})}{k!}.
\end{align}
\else
\begin{align}
&\Pr\!\,^{r}(K_1=k_1)=\mathbb{E}_{K}\left\{\Pr\!\,^{r}(K_1=k_1|K)\right\} \nonumber\\
\approx&\sum_{k=k_1}^{N}\binom{k}{k_1}(\apouta^{r})^{k_1}(1-\apouta^{r})^{k-k_1}\frac{\bar{m}^k\exp(-\bar{m})}{k!}\nonumber\\
&+\sum_{k=N+1}^{\infty}\binom{N}{k_1}(\apouta^{r})^{k_1}(1-\apouta^{r})^{k-k_1}\frac{\bar{m}^k\exp(-\bar{m})}{k!}.
\end{align}
\fi
\noindent Further simplifying the above equation, we arrive the result in Corollary~\ref{coro:distribution1}.
\end{proof}

\section*{Appendix E: Proof of Proposition~\ref{theo:channeloutage2}}\label{appendix:outage2}
\begin{proof}
In order to derive the average channel success probability for the relaying phase, let us first compute the conditional channel success. We assume that the typical BS is located at the origin. For an active aggregator that is associated with this typical BS, given the number of active channels for this aggregator $K_1$ and the total number of active aggregators associated with the BS $N_a$, the conditional channel success probability is given by
\ifCLASSOPTIONpeerreview
\begin{align}\label{eq:conditionoutage2}
\poutb(K_1,N_a)=\Pr\left(\frac{h'}{I_2}\geq2^{\frac{D K_1 N_a}{TW}}-1\right)=\sum\limits_{t=0}\limits^{m_2-1}\frac{(-s)^{t}}{t!}\left.\frac{d^t}{d s^t}\mathcal{M}_{I_{2}}(s)\right|_{s=m_2(2^{\frac{D K_1 N_a}{TW}}-1)}.
\end{align}
\else
\begin{align}\label{eq:conditionoutage2}
\poutb(K_1,N_a)&=\Pr\left(\frac{h'}{I_2}\geq2^{\frac{D K_1 N_a}{TW}}-1\right)\nonumber\\
&=\sum\limits_{t=0}\limits^{m_2-1}\frac{(-s)^{t}}{t!}\left.\frac{d^t}{d s^t}\mathcal{M}_{I_{2}}(s)\right|_{s=m_2\left(2^{\frac{D K_1 N_a}{TW}}-1\right)}.
\end{align}
\fi
\noindent where $h'$ is the fading gain on the desired link from aggregator to the typical BS which follows a gamma distribution with shape parameter $m_2$, $I_2=\sum_{y\subset\Phi_{a}^{\textrm{intf}}}g r_a^{\alpha}y^{-\alpha}$ is the aggregate interference from active aggregators occupying the same resource~\cite{Singh-2015}. Therein, $y$ denotes both the location and the interfering aggregators, and $r_a$ is the distance from the interfering aggregator to its associated BS where its PDF is $f(r_a)=2\pi \BSdensity r_a\exp(-2\pi\BSdensity r_a^2)$.

As we consider the orthogonal access for each aggregator within each BS, only one aggregator per cell can use the certain resource block and generate the interference to the typical BS. In general, the point process $\Phi_{a}^{\textrm{intf}}$ is a Poisson-Voronoi perturbed lattice~\cite{Blaszczyszyn-2015}. However, the consideration of such a point process leads to intractable results. Hence, similar to~\cite{Singh-2015,7136376}, we approximate the the location of interfering aggregators as a HPPP with density $(1-p_{\textrm{void}})\BSdensity$, where $p_{\textrm{void}}=\left(1+\frac{\adensity'}{3.5\BSdensity}\right)^{-3.5}$ is the void probability that accounts for the probability that BS has zero active aggregator~\cite{7248753}.

Based on the above approximation, the MGF of the interference at the relaying phase is given by
\ifCLASSOPTIONpeerreview
\begin{align}
\mathcal{M}_{I_2}(s)=&\mathbb{E}_{\Phi_{a}^{\textrm{intf}},g,r_a}\left\{\sum_{y\subset\Phi_{a}^{\textrm{intf}}}g r_a^{\alpha}y^{-\alpha}\right\}\approx\exp\left(-\mathbb{E}_{r_a,g}\left\{\int_{r_a}^{\infty}g r_a^{\alpha}r^{-\alpha}2\pi(1-p_{\textrm{void}})\BSdensity r \textup{d}r\right\}\right) \nonumber\\
=&\exp\left(-2\pi(1-p_{\textrm{void}})\BSdensity\int_{0}^{\infty}\int_{r_a}^{\infty}\frac{sr_a^{\alpha}r^{-\alpha}}{1+sr_a^{\alpha}r^{-\alpha}}r\textup{d}r f(r_a)\textup{d}r_a\right) \nonumber\\
=&\exp\left(-2\pi(1-p_{\textrm{void}})\BSdensity\int_{0}^{\infty}\int_{{\frac{1}{s}}^{\frac{1}{\alpha}}}^{\infty}\frac{s^{\frac{2}{\alpha}}z}{1+z^{\alpha}}r_a^{2}\textup{d}zf(r_a)\textup{d}r_a\right)\nonumber\\
=&\exp\left(-2\pi(1-p_{\textrm{void}})\BSdensity s^{\frac{2}{\alpha}}\frac{s^{1-\frac{\alpha}{2}}\,_2F_1\left[1,1-\frac{2}{\alpha},2-\frac{2}{\alpha},-s\right]}{\alpha-2}\mathbb{E}_{r_a}\left\{r_a^2\right\} \right) \nonumber\\
=&\exp\left(-2(1-p_{\textrm{void}})s\frac{\,_2F_1\left[1,1-\frac{2}{\alpha},2-\frac{2}{\alpha},-s\right]}{\alpha-2}\right),
\end{align}
\else
\begin{align}
&\mathcal{M}_{I_2}(s)=\mathbb{E}_{\Phi_{a}^{\textrm{intf}},g,r_a}\left\{\sum_{y\subset\Phi_{a}^{\textrm{intf}}}g r_a^{\alpha}y^{-\alpha}\right\}\nonumber\\
&\approx\exp\!\left(-\mathbb{E}_{r_a,g}\left\{\int_{r_a}^{\infty}g r_a^{\alpha}r^{-\alpha}2\pi(1-p_{\textrm{void}})\BSdensity r \textup{d}r\right\}\right) \nonumber\\
&=\exp\!\left(-2\pi(1\!-\!p_{\textrm{void}})\BSdensity\!\!\int_{0}^{\infty}\!\!\!\int_{r_a}^{\infty}\!\!\frac{sr_a^{\alpha}r^{-\alpha}}{1+sr_a^{\alpha}r^{-\alpha}}r\textup{d}r f(r_a)\textup{d}r_a\!\right) \nonumber\\
&=\exp\!\left(-2\pi(1\!-\!p_{\textrm{void}})\BSdensity\int_{0}^{\infty}\int_{{\frac{1}{s}}^{\frac{1}{\alpha}}}^{\infty}\frac{s^{\frac{2}{\alpha}}z}{1+z^{\alpha}}r_a^{2}\textup{d}zf(r_a)\textup{d}r_a\right)\nonumber\\
&=\exp\!\Bigg(-2\pi(1\!-\!p_{\textrm{void}})\BSdensity s^{\frac{2}{\alpha}}\mathbb{E}_{r_a}\left\{r_a^2\right\} \nonumber\\
&\quad\quad\quad\quad\times\frac{s^{1-\frac{\alpha}{2}}\,_2F_1\left[1,1-\frac{2}{\alpha},2-\frac{2}{\alpha},-s\right]}{\alpha-2}\Bigg) \nonumber\\
&=\exp\left(-2(1-p_{\textrm{void}})s\frac{\,_2F_1\left[1,1-\frac{2}{\alpha},2-\frac{2}{\alpha},-s\right]}{\alpha-2}\right),
\end{align}
\fi
\noindent where the second step is based on the approximation of $\Phi_{a}^{\textrm{intf}}$ being HPPP and then followed by the Campbell theorem~\cite{Haenggi-2012}. \redcom{Different from~\eqref{eq:MGFintf1} where the distance between the interfering MTD and the typical aggregator is in the range of $[0,\infty]$, in this case, $r$ ranges from $r_a$ to $\infty$ because of our considered aggregator association scheme. The aggregator is always associated to its nearest BS; hence the distance between an aggregator and its associated BS $r_a$ is always shorter than the distance between an aggregator and the interfered typical BS $r$. Otherwise, it violates the association scheme.} The third step comes from the fact that the fading gain on the interfering link follows an exponential distribution. The fourth step is based on the variable substitution, i.e., $z=rs^{\frac{1}{\alpha}}r_a$ and the last step comes from the fact that $\mathbb{E}_{r_a}\left\{r_a^2\right\}=\int_{0}^{\infty}r_a^22\pi \BSdensity r_a\exp(-2\pi\BSdensity r_a^2)\textup{d}r_a=\frac{1}{\pi\BSdensity}$.

The next step is to de-condition the above conditional channel success probability. The distribution of $K_1$ is already presented in Corollaries~\ref{coro:distribution1} and~\ref{coro:distribution1p} for different schemes. In terms of the distribution of $N_a$, there is no exact closed-form expression for this distribution. However, a simple approximation has been proposed in~\cite{Ferenc-2007}, which can provide sufficient accuracy for practical purposes~\cite{Zhong-2013} and has been widely adopted in the literature. The approximation for the distribution of the number of active aggregator associated with a BS is given by
\begin{align}\label{eq:Nadist}
\Pr\left(N_a=n_a\right)=\frac{3.5^{3.5}\Gamma\![n_a+3.5](\frac{\adensity'}{\BSdensity})^{n_a}}{\Gamma\![3.5]\Gamma\![n_a+1](3.5+\frac{\adensity'}{\BSdensity})^{n_a+3.5}}.
\end{align}

After averaging~\eqref{eq:conditionoutage2} over $N_a$ and $K_1$, we can arrive the result presented in Proposition~\ref{theo:channeloutage2}. Note that the extra term in denominator $\left(1-\Pr\left(N_a=0\right)\right)$ comes from the fact the derived average channel success probability is calculated for the link from the active aggregator to BS, which requires $K_1\geq 1$ and $N_a\geq1$.
\end{proof}

 \bibliographystyle{IEEEtran}

\begin{thebibliography}{10}
\providecommand{\url}[1]{#1}
\csname url@samestyle\endcsname
\providecommand{\newblock}{\relax}
\providecommand{\bibinfo}[2]{#2}
\providecommand{\BIBentrySTDinterwordspacing}{\spaceskip=0pt\relax}
\providecommand{\BIBentryALTinterwordstretchfactor}{4}
\providecommand{\BIBentryALTinterwordspacing}{\spaceskip=\fontdimen2\font plus
\BIBentryALTinterwordstretchfactor\fontdimen3\font minus
  \fontdimen4\font\relax}
\providecommand{\BIBforeignlanguage}[2]{{%
\expandafter\ifx\csname l@#1\endcsname\relax
\typeout{** WARNING: IEEEtran.bst: No hyphenation pattern has been}%
\typeout{** loaded for the language `#1'. Using the pattern for}%
\typeout{** the default language instead.}%
\else
\language=\csname l@#1\endcsname
\fi
#2}}
\providecommand{\BIBdecl}{\relax}
\BIBdecl

\bibitem{6231296}
K.~Zheng, F.~Hu, W.~Wang, W.~Xiang, and M.~Dohler, ``Radio resource allocation
  in {LTE}-advanced cellular networks with {M2M} communications,'' \emph{IEEE
  Commun. Mag.}, vol.~50, no.~7, pp. 184--192, Jul. 2012.

\bibitem{Bockelmann-2016}
C.~Bockelmann, N.~Pratas, H.~Nikopour, K.~Au, T.~Svensson, C.~Stefanovic,
  P.~Popovski, and A.~Dekorsy, ``Massive machine-type communications in {5G}:
  Physical and {MAC}-layer solutions,'' \emph{IEEE Commun. Mag.}, vol.~54,
  no.~9, pp. 59--65, Sep. 2016.

\bibitem{Laya-2014}
A.~Laya, L.~Alonso, and J.~Alonso-Zarate, ``Is the random access channel of
  {LTE} and {LTE-A} suitable for {M2M} communications? {A} survey of
  alternatives,'' \emph{IEEE Commun. Surveys Tuts.}, vol.~16, no.~1, pp. 4--16,
  First quarter 2014.

\bibitem{3gpp}
{3GPP TS 36.331 V10.50.0}, ``{Evolved Universal Terrestrial Radio Access
  (E-UTRA); Radio Resource Control (RRC)},'' Mar. 2012.

\bibitem{6720118}
T.~M. Lin, C.~H. Lee, J.~P. Cheng, and W.~T. Chen, ``{PRADA}: Prioritized
  random access with dynamic access barring for {MTC} in {3GPP LTE-A}
  networks,'' \emph{IEEE Trans. Veh. Technol.}, vol.~63, no.~5, pp. 2467--2472,
  Jun. 2014.

\bibitem{3gpp-2}
{3GPP TSG RAN WG2 $\#$71 R2-104662}, ``{MTC simulation results with specific
  solutions}.''\hskip 1em plus 0.5em minus 0.4em\relax ZTE, Madrid, Spain, Aug.
  2010.

\bibitem{6399192}
X.~Yang, A.~Fapojuwo, and E.~Egbogah, ``Performance analysis and parameter
  optimization of random access backoff algorithm in {LTE},'' in \emph{Proc.
  IEEE VTC-Fall}, Sep. 2012, pp. 1--5.

\bibitem{Dawy-2015}
Z.~Dawy, W.~Saad, A.~Ghosh, J.~G. Andrews, and E.~Yaacoub, ``Towards massive
  machine type cellular communications,'' \emph{IEEE Wireless Commun.},
  vol.~24, no.~1, pp. 120--128, Feb. 2017.

\bibitem{6175028}
C.~Y. Ho and C.~Y. Huang, ``Energy-saving massive access control and resource
  allocation schemes for {M2M} communications in {OFDMA} cellular networks,''
  \emph{IEEE Wireless Commun. Lett.}, vol.~1, no.~3, pp. 209--212, Jun. 2012.

\bibitem{Miao-2016}
G.~Miao, A.~Azari, and T.~Hwang, ``E$^2$-{MAC}: Energy efficient medium access
  for massive {M2M} communications,'' \emph{{IEEE} Trans. Commun.}, vol.~64,
  no.~11, pp. 4720--4735, Nov. 2016.

\bibitem{Pratas-2015}
N.~K. Pratas and P.~Popovski, ``Zero-outage cellular downlink with fixed-rate
  {D2D} underlay,'' \emph{{IEEE} Trans. Wireless Commun.}, vol.~14, no.~7, pp.
  3533--3543, Jul. 2015.

\bibitem{7248779}
G.~Rigazzi, N.~K. Pratas, P.~Popovski, and R.~Fantacci, ``Aggregation and
  trunking of {M2M} traffic via {D2D} connections,'' in \emph{Proc. IEEE ICC},
  Jun. 2015, pp. 2973--2978.

\bibitem{Kwon-2013}
T.~Kwon and J.~M. Cioffi, ``Random deployment of data collectors for serving
  randomly-located sensors,'' \emph{{IEEE} Trans. Wireless Commun.}, vol.~12,
  no.~6, pp. 2556--2565, Jun. 2013.

\bibitem{halim-2016}
Z.~Zhao, S.~Szyszkowicz, T.~Beitelmal, and H.~Yanikomeroglu, ``Spatial
  clustering in slotted {ALOHA} two-hop random access for machine type
  communication,'' in \emph{Proc. IEEE Globecom}, Dec. 2016, pp. 1--6.

\bibitem{Malak-2016}
D.~Malak, H.~S. Dhillon, and J.~G. Andrews, ``Optimizing data aggregation for
  uplink machine-to-machine communication networks,'' \emph{{IEEE} Trans.
  Commun.}, vol.~64, no.~3, pp. 1274--1290, Mar. 2016.

\bibitem{6687313}
N.~Abu-Ali, A.~E.~M. Taha, M.~Salah, and H.~Hassanein, ``Uplink scheduling in
  {LTE} and {LTE}-{A}dvanced: Tutorial, survey and evaluation framework,''
  \emph{IEEE Commun. Surveys Tuts.}, vol.~16, no.~3, pp. 1239--1265, Third
  quarter 2014.

\bibitem{7725919}
M.~Shirvanimoghaddam, M.~Dohler, and S.~J. Johnson, ``Massive multiple access
  based on superposition raptor codes for cellular {M2M} communications,''
  \emph{{IEEE} Trans. Wireless Commun.}, vol.~16, no.~1, pp. 307--319, Jan.
  2017.

\bibitem{6504002}
C.-H. Chang and H.-Y. Hsieh, ``Not every bit counts: A resource allocation
  problem for data gathering in machine-to-machine communications,'' in
  \emph{Proc. IEEE Globecom}, Dec. 2012, pp. 5537--5543.

\bibitem{6477828}
A.~G. Gotsis, A.~S. Lioumpas, and A.~Alexiou, ``Evolution of packet scheduling
  for machine-type communications over {LTE}: Algorithmic design and
  performance analysis,'' in \emph{Proc. IEEE Globecom Workshops}, Dec. 2012,
  pp. 1620--1625.

\bibitem{7417719}
S.~Hamdoun, A.~Rachedi, and Y.~Ghamri-Doudane, ``Radio resource sharing for
  {MTC} in {LTE-A}: An interference-aware bipartite graph approach,'' in
  \emph{Proc. IEEE Globecom}, Dec. 2015, pp. 1--7.

\bibitem{Kumar-2016}
\BIBentryALTinterwordspacing
A.~Kumar, A.~Abdelhadi, and C.~Clancy, ``A delay optimal {MAC} and packet
  scheduler for heterogeneous {M2M} uplink,'' 2016, submitted. [Online].
  Available: \url{http://arxiv.org/abs/1606.06692}
\BIBentrySTDinterwordspacing

\bibitem{Haenggi-2012}
M.~Haenggi, \emph{Stochastic Geometry for Wireless Networks}.\hskip 1em plus
  0.5em minus 0.4em\relax Cambridge University Press, 2012.

\bibitem{jeffrey-2011}
J.~G. Andrews, F.~Baccelli, and R.~K. Ganti, ``A tractable approach to coverage
  and rate in cellular networks,'' \emph{{IEEE} Trans. Commun.}, vol.~59,
  no.~11, pp. 3122--3134, Nov. 2011.

\bibitem{Guo-2016}
J.~Guo, S.~Durrani, X.~Zhou, and H.~Yanikomeroglu, ``Device-to-device
  communication underlaying a finite cellular network region,'' \emph{{IEEE}
  Trans. Wireless Commun.}, vol.~16, no.~1, pp. 332--347, Jan. 2017.

\bibitem{Azari-2016}
A.~Azari, ``Energy efficient scheduling and grouping for machine-type
  communications over cellular networks,'' \emph{Ad Hoc Netw.}, vol.~43, pp.
  16--29, Jun. 2016.

\bibitem{Dhillon-2014}
H.~S. Dhillon, H.~Huang, H.~Viswanathan, and R.~A. Valenzuela, ``Fundamentals
  of throughput maximization with random arrivals for {M2M} communications,''
  \emph{{IEEE} Trans. Commun.}, vol.~62, no.~11, Nov. 2014.

\bibitem{Singh-2015}
S.~Singh, X.~Zhang, and J.~G. Andrews, ``Joint rate and {SINR} coverage
  analysis for decoupled uplink-downlink biased cell associations in
  {HetNet}s,'' \emph{{IEEE} Trans. Wireless Commun.}, vol.~14, no.~10, pp.
  5360--5373, Oct. 2015.

\bibitem{Mohammad-2016}
\BIBentryALTinterwordspacing
M.~Gharbieh, H.~ElSawy, A.~Bader, and M.-S. Alouini, ``Spatiotemporal
  stochastic modeling of {IoT} enabled cellular networks: Scalability and
  stability analysis,'' 2016, submitted. [Online]. Available:
  \url{https://arxiv.org/abs/1609.05384}
\BIBentrySTDinterwordspacing

\bibitem{Gupta-1960}
S.~S. Gupta, ``Order statistics from the gamma distribution,''
  \emph{Technometrics}, vol.~2, no.~2, pp. 243--262, May 1960.

\bibitem{7136376}
H.~Xu, C.~Yongyu, T.~Sun, D.~Zheng, and D.~Yang, ``Modeling and analysis of
  uplink-downlink relationship in heterogeneous cellular network,'' in
  \emph{Proc. IEEE PIMRC}, Sep. 2014, pp. 1338--1342.

\bibitem{6516885}
T.~D. Novlan, H.~S. Dhillon, and J.~G. Andrews, ``Analytical modeling of uplink
  cellular networks,'' \emph{{IEEE} Trans. Wireless Commun.}, vol.~12, no.~6,
  pp. 2669--2679, Jun. 2013.

\bibitem{Mahyar-2017}
\BIBentryALTinterwordspacing
M.~Shirvanimoghaddam, M.~Dohler, and S.~Johnson, ``Massive non-orthogonal
  multiple access for cellular {IoT}: Potentials and limitations,'' \emph{IEEE
  Commun. Mag.}, 2017, (to appear). [Online]. Available:
  \url{https://arxiv.org/abs/1612.00552}
\BIBentrySTDinterwordspacing

\bibitem{David-1970}
H.~A. David and H.~N. Nagaraja, \emph{Order Statistics}.\hskip 1em plus 0.5em
  minus 0.4em\relax New York: John Wiley and Sons, 1970.

\bibitem{Pelaez-1951}
J.~Gil-Pelaez, ``Note on the inversion theorem,'' \emph{Biometrika}, vol.~38,
  no. 3/4, pp. 481--482, Dec. 1951.

\bibitem{Renzo-2014}
M.~D. Renzo and P.~Guan, ``Stochastic geometry modeling of coverage and rate of
  cellular networks using the {G}il-{P}elaez inversion theorem,'' \emph{{IEEE}
  Commun. Lett.}, vol.~18, no.~9, pp. 1575--1578, Sep. 2014.

\bibitem{Blaszczyszyn-2015}
B.~Blaszczyszyn and D.~Yogeshwaran, ``Clustering comparison of point processes,
  with applications to random geometric models,'' in \emph{Stochastic Geometry,
  Spatial Statistics and Random Fields}.\hskip 1em plus 0.5em minus 0.4em\relax
  Springer International Publishing, 2015.

\bibitem{7248753}
C.~H. Liu and L.~C. Wang, ``Random cell association and void probability in
  {P}oisson-distributed cellular networks,'' in \emph{Proc. IEEE ICC}, Jun.
  2015, pp. 2816--2821.

\bibitem{Ferenc-2007}
J.~Ferenc and Z.~Neda, ``On the size distribution of {P}oisson voronoi cells,''
  \emph{Phys. A}, vol. 385, no.~2, pp. 518--526, 2007.

\bibitem{Zhong-2013}
Y.~Zhong and W.~Zhang, ``Multi-channel hybrid access femtocells: A stochastic
  geometric analysis,'' \emph{{IEEE} Trans. Commun.}, vol.~61, no.~7, pp.
  3016--3026, Jul. 2013.

\end{thebibliography}

\end{document}